\documentclass[twocolumn]{IEEEtran}
\usepackage[mathscr]{eucal}
\usepackage[cmex10]{amsmath}
\usepackage{epsfig,epsf,psfrag}
\usepackage{amssymb,amsmath,amsthm,amsfonts,latexsym}
\usepackage{amsmath,graphicx,bm,xcolor,url,overpic}
\usepackage{fixltx2e}
\usepackage{array}
\usepackage{verbatim}
\usepackage{bm}
\usepackage{algorithmic}
\usepackage{algorithm}
\usepackage{verbatim}
\usepackage{textcomp}
\usepackage{mathrsfs}
\usepackage{epstopdf}

\newcommand{\openone}{\leavevmode\hbox{\small1\normalsize\kern-.33em1}}

\catcode`~=11 \def\UrlSpecials{\do\~{\kern -.15em\lower .7ex\hbox{~}\kern .04em}} \catcode`~=13 

\allowdisplaybreaks[4]

\newcommand{\nn}{\nonumber}

\newcommand{\calA}{\mathcal{A}}
\newcommand{\calB}{\mathcal{B}}
\newcommand{\calC}{\mathcal{C}}
\newcommand{\calD}{\mathcal{D}}

\newcommand{\calF}{\mathcal{F}}

\newcommand{\calL}{\mathcal{L}}
\newcommand{\calM}{\mathcal{M}}

\newcommand{\calP}{\mathcal{P}}
\newcommand{\calQ}{\mathcal{Q}}
\newcommand{\calR}{\mathcal{R}}

\newcommand{\calT}{\mathcal{T}}
\newcommand{\calU}{\mathcal{U}}

\newcommand{\calX}{\mathcal{X}}
\newcommand{\calY}{\mathcal{Y}}
\newcommand{\calZ}{\mathcal{Z}}

\newcommand{\bt}{\mathbf{t}}
\newcommand{\bT}{\mathbf{T}}

\newcommand{\bx}{\mathbf{x}}
\newcommand{\bX}{\mathbf{X}}
\newcommand{\by}{\mathbf{y}}
\newcommand{\bY}{\mathbf{Y}}
\newcommand{\bz}{\mathbf{z}}
\newcommand{\bZ}{\mathbf{Z}}

\newcommand{\rmc}{\mathrm{c}}

\newcommand{\rme}{\mathrm{e}}

\newcommand{\rmi}{\mathrm{i}}

\newcommand{\rmP}{\mathrm{P}}

\newcommand{\rmV}{\mathrm{V}}

\newcommand{\bbN}{\mathbb{N}}

\newcommand{\bbR}{\mathbb{R}}

\DeclareMathAlphabet{\mathbsf}{OT1}{cmss}{bx}{n}
\DeclareMathAlphabet{\mathssf}{OT1}{cmss}{m}{sl}

\DeclareSymbolFont{bsfletters}{OT1}{cmss}{bx}{n}  
\DeclareSymbolFont{ssfletters}{OT1}{cmss}{m}{n}
\DeclareMathSymbol{\bsfGamma}{0}{bsfletters}{'000}
\DeclareMathSymbol{\ssfGamma}{0}{ssfletters}{'000}
\DeclareMathSymbol{\bsfDelta}{0}{bsfletters}{'001}
\DeclareMathSymbol{\ssfDelta}{0}{ssfletters}{'001}
\DeclareMathSymbol{\bsfTheta}{0}{bsfletters}{'002}
\DeclareMathSymbol{\ssfTheta}{0}{ssfletters}{'002}
\DeclareMathSymbol{\bsfLambda}{0}{bsfletters}{'003}
\DeclareMathSymbol{\ssfLambda}{0}{ssfletters}{'003}
\DeclareMathSymbol{\bsfXi}{0}{bsfletters}{'004}
\DeclareMathSymbol{\ssfXi}{0}{ssfletters}{'004}
\DeclareMathSymbol{\bsfPi}{0}{bsfletters}{'005}
\DeclareMathSymbol{\ssfPi}{0}{ssfletters}{'005}
\DeclareMathSymbol{\bsfSigma}{0}{bsfletters}{'006}
\DeclareMathSymbol{\ssfSigma}{0}{ssfletters}{'006}
\DeclareMathSymbol{\bsfUpsilon}{0}{bsfletters}{'007}
\DeclareMathSymbol{\ssfUpsilon}{0}{ssfletters}{'007}
\DeclareMathSymbol{\bsfPhi}{0}{bsfletters}{'010}
\DeclareMathSymbol{\ssfPhi}{0}{ssfletters}{'010}
\DeclareMathSymbol{\bsfPsi}{0}{bsfletters}{'011}
\DeclareMathSymbol{\ssfPsi}{0}{ssfletters}{'011}
\DeclareMathSymbol{\bsfOmega}{0}{bsfletters}{'012}
\DeclareMathSymbol{\ssfOmega}{0}{ssfletters}{'012}

\newcommand{\tilC}{\tilde{C}}

\newcommand{\tilF}{\tilde{F}}

\newcommand{\hatW}{\hat{W}}

\newcommand{\hatx}{\hat{x}}
\newcommand{\hatX}{\hat{X}}

\newtheorem{theorem}{Theorem} 
\newtheorem{lemma}[theorem]{Lemma}

\newtheorem{definition}{Definition}

\usepackage{cite}
\usepackage[ colorlinks = true,
             linkcolor = blue,
             urlcolor  = blue,
             citecolor = red,
             anchorcolor = green,]{hyperref}
\newcommand{\blue}{\textcolor{blue}}
\begin{document}

\title{Exponential Strong Converse for Content Identification with Lossy Recovery}
\author{\IEEEauthorblockN{Lin Zhou,~\IEEEmembership{Student Member, IEEE}, Vincent Y.\ F.\ Tan,~\IEEEmembership{Senior Member, IEEE}, \\  Lei Yu, Mehul Motani,~\IEEEmembership{Fellow, IEEE}} \thanks{The authors (emails:  lzhou@u.nus.edu,~\{vtan, leiyu, motani\}@nus.edu.sg) are  with the Department of Electrical and Computer Engineering, National University of Singapore (NUS). V.~Y.~F.~Tan is also with the Department of Mathematics, NUS. } \thanks{This paper was presented in part at the 2017 International Symposium on Information Theory  (ISIT) in Aachen, Germany~\cite{ZhouTanMotaniContent}.}
}
\maketitle

\begin{abstract}
\blue{We revisit the high-dimensional content identification with lossy recovery problem (Tuncel and G\"und\"uz, 2014) and establish an exponential strong converse theorem.} As a corollary of the exponential strong converse theorem, we derive an upper bound on the joint identification-error and excess-distortion exponent for the problem. Our main results can be specialized to the biometrical identification problem~(Willems, 2003) and the content identification problem~(Tuncel, 2009) since these two problems are both special cases of the content identification with lossy recovery problem. We leverage the information spectrum method introduced by Oohama  and adapt the strong converse techniques therein to be applicable to the problem at hand.
\end{abstract}

\begin{IEEEkeywords}
Content identification, Lossy source coding, Biometrical identification, Exponential strong converse, Information spectrum method
\end{IEEEkeywords}

\vspace{-.05in}
\section{Introduction}

Have you ever wondered about the identity of a song after hearing only a short snippet? With limited information, it is sometimes difficult to identify the song or a distorted version of it, yet not impossible. In fact, there is an app called {\em Shazam} that does precisely this. There are three distinct steps in the process of identifying the song, namely, the enrollment phase, the identification phase and the lossy recovery phase (see Figure \ref{systemmodel}). In the enrollment phase, the database of songs is sought; in the identification phase, we would like to infer certain details about the song; and finally, in the recovery phase, we hope to recover (at least) a lossy version of the song. An information-theoretic model was put forth by Tuncel and G\"und\"uz~\cite{tuncel2014idenlossy} and they called this model the (high-dimensional) content identification problem with lossy recovery. This model is also applicable to other situations such as fingerprint identification~\cite{naini2014fingerprint} and video identification~\cite{yu2013videocid}. However, \cite{tuncel2014idenlossy} only established a weak converse.  In this paper, we revisit the content identification problem with lossy recovery and derive an exponential strong converse theorem.

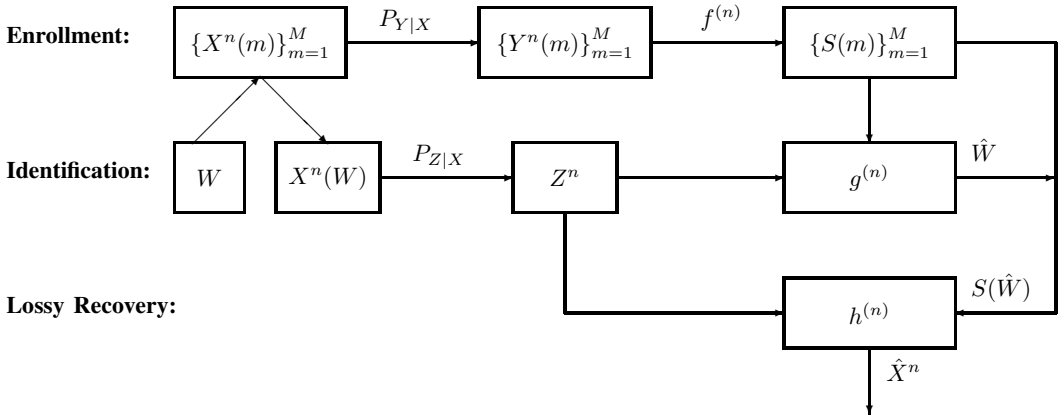
\begin{figure*}[t]
\centering
\setlength{\unitlength}{0.5cm}
\scalebox{0.9}{
\begin{picture}(30,13)
\linethickness{1pt}
\put(0,12){\makebox{\textbf{Enrollment:}}}
\put(5,11){\framebox(5,2){$\left\{X^n(m)\right\}_{m=1}^M$}}
\put(10,12){\vector(1,0){4}}
\put(11,12.5){\makebox{$P_{Y|X}$}}
\put(14,11){\framebox(5,2){$\left\{Y^n(m)\right\}_{m=1}^M$}}
\put(19,12){\vector(1,0){4}}
\put(20.5,12.5){\makebox{$f^{(n)}$}}
\put(23,11){\framebox(5,2){$\left\{S(m)\right\}_{m=1}^M$}}
\put(0,8){\makebox{\textbf{Identification:}}}
\put(5,7){\framebox(2,2){$W$}}
\put(5.5,9){\vector(1,1){2}}
\put(7.5,11){\vector(1,-1){2}}
\put(8,7){\framebox(3,2){$X^n(W)$}}
\put(11,8){\vector(1,0){4}}
\put(12,8.5){\makebox{$P_{Z|X}$}}
\put(15,7){\framebox(3,2){$Z^n$}}
\put(18,8){\vector(1,0){5}}
\put(23,7){\framebox(5,2){$g^{(n)}$}}
\put(25.5,11){\vector(0,-1){2}}
\put(28,8){\vector(1,0){3}}
\put(28.5,8.5){\makebox{$\hatW$}}
\put(0,4){\makebox{\textbf{Lossy Recovery:}}}
\put(16.5,7){\line(0,-1){3}}
\put(16.5,4){\vector(1,0){6.5}}
\put(23,3){\framebox(5,2){$h^{(n)}$}}
\put(28,12){\line(1,0){3}}
\put(31,12){\line(0,-1){8}}
\put(31,4){\vector(-1,0){3}}
\put(28.5,4.5){\makebox{$S(\hatW)$}}
\put(25.5,3){\vector(0,-1){2}}
\put(26,2){\makebox{$\hatX^n$}}
\end{picture}}
\caption{Content identification with lossy recovery~\cite{tuncel2014idenlossy}.}
\label{systemmodel}
\end{figure*}

\vspace{-.05in}
\subsection{Related Works}

The most related works are \cite{tuncel2009capacity} and \cite{tuncel2014idenlossy}. In \cite{tuncel2009capacity}, Tuncel characterizes the achievable rate region of the content identification problem. In \cite{tuncel2014idenlossy}, Tuncel and G\"und\"uz characterized the rate-distortion region for content identification problem with lossy recovery. Other (non-exhausting) works on the content identification problem are summarized as follows. Willems {\em et al.}~\cite{willems2003capacity} initiated the study of the content identification problem by characterizing the capacity of a biometrical identification problem. Dasarathy and Draper~\cite{dasarathy2014upper} derived upper and lower bounds on the error exponent of the content identification system where they assume the DMC $P_{Y|X}$ is a noiseless channel. Recently, Merhav~\cite{merhav2016reliability} refined the result in~\cite{dasarathy2014upper} by proposing a universal achievability scheme and showing that the scheme achieves the optimal exponent given by maximum likelihood decoding. Furthermore, Yachongka and Yagi~\cite{yagi2016} established the strong converse theorem for the biometrical identification problem. We remark that Yachongka and Yagi used Arimoto's strong converse technique~\cite{Arimoto}  which is different from the information spectrum method adopted in this paper. The main result of~\cite{yagi2016} is recovered as a by-product of our main result. Other works on content identification include  \cite{nani2012ci,nani2013adaboost,moulin2010ci,farhadzadeh2011ci,gunduz2009ci2database,farhadzadeh2014twostage,tuncel2012recognitionci}.

We also summarize the works by Oohama on using the information spectrum method to derive exponential strong converse theorems for several network information theory problems. In \cite{oohama2015wak,oohama2016new,oohama2016wynerziv}, Oohama derived exponential strong converses for \blue{the lossless source coding problem  with one-helper (i.e., the Wyner-Ahlswede-K\"orner (WAK) problem)~\cite{ahlswede1975,wyner1975}}, the asymmetric broadcast channel problem~\cite{korner1977abc}, and the Wyner-Ziv problem~\cite{wyner1976rate} respectively. Furthermore, Oohama's information spectrum method was also used recently by Yu and Tan~\cite{yu2017wyner} to derive an exponential strong converse theorem for Wyner's common information problem under the total variation distance measure~\cite{wyner1975ci}.

\vspace{-.05in}
\subsection{Main Contributions and Challenges}

For the content identification problem with lossy recovery, we first present a non-asymptotic converse bound. Invoking the non-asymptotic converse bound, we establish an upper bound on the probability of correct decoding in both the content identification index and the feature vector. By correct decoding of the feature vector, we mean that the reproduced feature vector is within certain distortion level under a distortion measure. Furthermore, we show that the probability of correct decoding decays exponentially fast to zero if the rate-distortion tuple falls outside the rate-distortion region by Tuncel and G\"und\"uz in \cite{tuncel2014idenlossy}. Hence, we establish an exponential strong converse theorem for the current problem. As a corollary, we derive an upper bound on the joint identification-error and excess-distortion exponent. Our results can be specialized to the biometrical identification problem~\cite{willems2003capacity} and the content identification problem~\cite{tuncel2009capacity}. In particular, for the biometrical identification problem, we derive the moderate deviations constant and the second-order coding rate.

In the rest of this subsection, we discuss the main challenges in establishing a strong converse theorem for the content identification problem with lossy recovery. First, we need to identify the correct form of the auxiliary random variables. As can be seen in the proofs in Section \ref{proofmainresult:idlossy}, the auxiliary random variables we choose are different from those in the weak converse proof~\cite{tuncel2014idenlossy}. If we choose the auxiliary variable as in the weak converse proof, we cannot establish the exponential strong converse result.

Second, the content identification problem with lossy recovery involves three phases: the enrollment phase, the identification phase and the lossy recovery phase. It is challenging to unify the analyses in different stages since the same auxiliary random variables are shared in all phases. Hence, we adopt ideas from \cite{oohama2016wynerziv,oohama2016new} which established strong converse theorems for the Wyner-Ziv problem and the degraded broadcast channel respectively. 

Third, in the identification phase, we need to use the whole random codebook and a noisy version of a certain feature vector to estimate the index of the feature vector (the identification index). This is very different from traditional channel coding and source coding problems. In source coding problems, we have only the codeword of a source sequence to decode while in channel coding problems, we have only the channel output for a particular message to decode. Hence, techniques like the image size characterization~\cite{csiszar2011information} and the perturbation approach~\cite{wei2009strong} are probably insufficient to establish a strong converse theorem for the current problem. As explained above, we adapt Oohama's strong converse techniques to deal with these challenges.

\vspace{-0.3cm}
\subsection{Organization of the Paper}
The rest of the paper is organized as follows. In Section \ref{sec:pformid}, we set up the notation, formulate the content identification problem with lossy recovery and recapitulate the existing results concerning the rate-distortion region. In Section \ref{sec:asymptotics}, we first present a non-asymptotic upper bound on the probability of correct decoding and then claim the exponential strong converse by studying the properties of the bound. As a corollary, we derive an upper bound on the joint identification-error and excess-distortion exponent. Our main results can be specialized to the content identification~\cite{tuncel2009capacity} and the biometrical identification~\cite{willems2003capacity} problems. 
Furthermore, for the biometrical identification problem, we derive the moderate deviations constant  and the second-order coding rate. The proof of our main result is presented in Section \ref{proofmainresult:idlossy}. Finally, we conclude the paper in Section \ref{sec:conclusion}. For seamless presentation of results, the proofs of all supporting lemmas are deferred to the appendices.

\section{Problem Formulation and Existing Results}
{\em Notation:}
Random variables and their realizations are in upper (e.g.,~$X$) and lower case (e.g.,\ $x$) respectively. Sets are denoted in calligraphic font (e.g.,\ $\mathcal{X}$). We use $\calX^{\mathrm{c}}$ to denote the complement of $\calX$ and $X^n:=(X_1,\ldots,X_n)$ is random vector of length $n$. We use $\mathbb{R}_+$ and $\bbN$ to denote the set of positive real numbers and integers respectively. Given a real number $a\in[0,1]$, we use $\bar{a}:=1-a$. Given two integers $a$ and $b$, we use $[a:b]$ to denote the set of all integers between $a$ and $b$. For quantities such as entropy and mutual information, we follow the notation in \cite{csiszar2011information}. The set of all probability distributions on $\calX$ is denoted as $\calP(\calX)$ and the set of all conditional probability distributions from $\calX$ to $\calY$ is denoted as $\calP(\calY|\calX)$.

\label{sec:pformid}
\subsection{Problem Formulation}
Let the random variables $(X,Y,Z,\hatX)$ take values in finite alphabets $\calX$, $\calY$, $\calZ$ and $\hat{\calX}$ respectively. Let $\calM:=\{1,\ldots,M\}$ and $\calL:=\{1,\ldots,L\}$. Let $d:\calX\times\hat{\calX}\to[0,\infty)$ be the distortion measure and let the distortion between $X^n$ and $\hatX^n$ be defined as $d(X^n,\hatX^n):=\frac{1}{n}\sum_{i=1}^n d(X_i,\hatX_i)$. Let the maximum distortion between $x\in\calX$ and $\hatx\in\hat{\calX}$ be $d^+$, i.e., $d^+:=\max_{x,\hatx}d(x,\hatx)$. Assume that each of the feature vectors $\{X^n(m)\}_{m\in\calM}$ is generated i.i.d.\ according to $P_X^n$. The content identification problem with lossy recovery is divided into three phases: the enrollment phase, the identification phase and the lossy recovery phase.  See Figure \ref{systemmodel}.

In the enrollment phase, for each $m\in\calM$, the noisy version $Y^n(m)$ of each feature vector $X^n(m)$ is observed, where $Y^n(m)$ is the output of passing $X^n(m)$ through a DMC with transition matrix $P_{Y|X}$ for $m\in\calM$, i.e.,
\begin{align}
P_{Y^n|X^n}(Y^n(m)|X^n(m))=\prod_{i=1}^n P_{Y|X}^n(Y_i(m)|X_i(m)).
\end{align}
 Subsequently, the observed noisy version of the feature vectors are compressed before stored in the database using a deterministic function
\begin{align}
f^{(n)}:\calY^n\to\calL:=\{1,\ldots,L\}.
\end{align} 
For convenience, let $S(m)=f(Y^n(m))$ for all $m\in\calM$.

In the identification phase, we are given an index $W$ which is uniformly generated  from the set $\calM$ and independent of $\{X^n(m),Y^n(m),S(m)\}_{m\in\calM}$. The index $W$  is unknown to the database users. Given $W$, database users observe $Z^n$, which is the output of passing the feature vector $X^n(W)$ through a DMC with transition matrix $P_{Z|X}$, i.e.,
\begin{align}
P_{Z|X}^n(Z^n(W)|X^n(W))=\prod_{i=1}^nP_{Z|X}(Z_i(W)|X_i(W)).
\end{align}
Note that $Z^n-X^n(W)-Y^n(W)$ forms a Markov chain. The user aims to identify $W$ using $Z^n$ and the compressed codebook $\{S(m)\}_{m\in\calM}$ using the following deterministic identification function:
\begin{align}
g^{(n)}:\calL^M\times\calZ^n\to\calM.
\end{align}
Let $\hatW:=g^{(n)}(S(1),\ldots,S(M),Z^n)$ be the estimate of the user. Given the deterministic decoding function $g^{(n)}$, we can define the following disjoint decoding regions 
\begin{align}
\nn&\calD(S(1),\ldots,S(M),W)\\*
&:=\{z^n:g^{(n)}(S(1),\ldots,S(M),z^n)=W\}.
\end{align}

Finally, in the lossy recovery phase, we need to reproduce the feature vector $X^n(W)$ in a lossy manner using $Z^n$ and $S(\hatW)$ with a deterministic function 
\begin{align}
h^{(n)}:\calL\times\calZ^n\to\hatX^n.
\end{align}
Let $\hatX^n=h(S(\hatW),Z^n)$ be the reproduced feature vector. Define the joint identification-error and excess-distortion probability as follows:
\begin{align}
\nn&\rmP_{\rm{e}}^{(n)}(f^{(n)},g^{(n)},h^{(n)},D)\\
&:=\Pr\big\{\hatW\neq W~\textrm{or}~d(X^n(W),\hatX^n)>D\big\}\\
\nn&=\prod_{m=1}^M\sum_{x^n(m), y^n(m), s(m)}   \Big(P_X^n(x^n(m))P_{Y|X}^n(y^n(m)|x^n(m))\\*
\nn&\qquad\qquad\qquad\qquad\times  1\Big\{s(m)=f^{(n)}(y^n(m))\Big\}\Big)\\*
\nn&\qquad\times \sum_{w=1}^M\frac{1}{M}\bigg(\sum_{\substack{(z^n,\hatx^n):~d(x^n(w),\hatx^n)>D\\\mathrm{or}~z^n\notin\calD(s(1),\ldots,s(M),w)}} P_{Z|X}^n(z^n|x^n(w))\\*
&\qquad\qquad\qquad\qquad\qquad\times 1\Big\{\hatx^n=h^{(n)}(s(w),z^n)\Big\}\bigg).\label{def:pejointerrorexcessdis}
\end{align}
Note that in \eqref{def:pejointerrorexcessdis} there are three sources of randomness: i) the randomness of feature vectors $x^n(m)\in\calX^n$ for each $m\in\calM$ in the enrollment phase; ii) the randomness of $w\in\calM$ in identification phase; iii) the randomness $y^n(m)\in\calY^n$ ($m\in\calM$) and $z^n(w)\in\calZ^n$  due to the two DMCs.

Throughout the paper, we will consider the source distribution being $P_X$ and the two DMCs with transition matrices $P_{Y|X}$ and $P_{Z|X}$. We will use $P_{XYZ}$ to denote $P_X\times P_{Y|X}\times P_{Z|X}$. In all the definitions, we will omit the dependence on distributions $P_X$, $P_{Y|X}$, and $P_{Z|X}$ for simplicity.

\subsection{Existing Results}
First, we define  the rate-distortion region.
\begin{definition}
\label{def:achratedistortion}
A rate-distortion triple $(R^{\rmi},R^{\rmc},D)$ is said to be $\varepsilon$-achievable if there exists a sequence encoding-decoding-reproduction functions $(f^{(n)},g^{(n)},h^{(n)})$ such that
\begin{align}
\liminf_{n\to\infty}\frac{1}{n}\log M&\geq R^{\rmi},\label{achdefRi}\\
\limsup_{n\to\infty}\frac{1}{n}\log L&\leq R^{\rmc},\label{achdefrc}\\
\limsup_{n\to\infty}\rmP_\rme^{(n)}(f^{(n)},g^{(n)},h^{(n)},D)&\leq \varepsilon.
\end{align}
The closure of all $\varepsilon$-achievable rate-distortion tuples is called the $\varepsilon$-rate-distortion region and denoted as $\calR(\varepsilon)$.
\end{definition}

Let
\begin{align}
\calR
&=\bigcap_{\varepsilon\in(0,1)}\calR(\varepsilon)\label{def:calr}.
\end{align}

In the following, we recall the rate-distortion region by Tuncel and G\"und\"uz~\cite[Theorem 1]{tuncel2014idenlossy}. We remark that their rate-distortion region appears to be identical with $\calR$ although it was derived under the average distortion criterion. 

Let $U$ be a random variable taking values in the alphabet $\calU$. Define a set of joint distributions on $\calX\! \times\! \calY\! \times\! \calZ\! \times\! \calU\! \times\! \hat{\calX}$ as
\begin{align}
\calP^*
\nn&:=\Big\{Q_{XYZU\hatX}:~|\calU|\leq |\calY|+2,~Z-X-Y-U,\\
\nn&\qquad\quad Q_X=P_X,~Q_{Y|X}=P_{Y|X},~Q_{Z|X}=P_{Z|X},\\
&\qquad\quad \hatX=\phi(U,Z)~\textrm{for~some}~\phi:\calU\times\calZ\to \hat{\calX}\Big\}\label{def:calPlossystar}.
\end{align}
Given $Q_{XYZU\hatX}$, let
\begin{align}
\!\!\!\!\calR(Q_{XYZU\hatX})
\nn=\Big\{(R^{\rmi},R^{\rmc},D):R^{\rmi}&\leq I(Q_U,Q_{Z|U})\\*
\nn\qquad R^{\rmc}-R^{\rmi}&\geq I(Q_{U|Z},Q_{Y|UZ}|Q_Z)\\*
\qquad D&\geq \mathbb{E}_{Q_{X\hatX}}\![d(X,\hatX)]\!\Big\},
\end{align}
and let
\begin{align}
\calR^*&:=\bigcup_{Q_{XYZU\hatX}\in\calP^*}\calR(Q_{XYZU\hatX})\label{def:calRstar}.
\end{align}
\begin{theorem}
\label{tuncelgidlossy}
The rate-distortion region $\calR$ satisfies
\begin{align}
\calR&=\calR^*.
\end{align}
\end{theorem}

\section{Main Results: Exponential Strong Converse Theorem}
\label{sec:asymptotics}
\subsection{Preliminaries}
In this subsection, we present some definitions and a key lemma in order to be able to succinctly state the exponential strong converse theorem in Section \ref{sec:esconverse}.

\blue{Recall that $\bar{a}=1-a$ for $a\in[0,1]$.} Let 
\begin{align}
\calQ&:=\big\{Q_{XYZU\hatX}:|\calU|\leq |\calX||\calY||\calZ||\hat{\calX}|\big\}\label{def:calQlossy}.
\end{align}

\blue{
Given $(\alpha,\theta)\in\bbR_+^2$, $(\mu,\beta)\in[0,1]^2$ and a distribution $Q_{XYZU\hatX}\in \calQ$, define the following linear combination of the log-likelihood ratios
\begin{align}
\nn&\omega^{(\alpha,\mu,\beta)}_{Q_{XYZU\hatX}}(x,y,z,\hatx|u)\\*
\nn&:=\log\frac{Q_Y(y)}{P_Y(y)}+\log\frac{Q_{Z|YU}(z|y,u)}{P_{Z|Y}(x|y)}+\log\frac{Q_{X|YZU}(x|y,z,u)}{P_{X|YZ}(x|y,z)}\\*
\nn&\quad+\log\frac{Q_{XY|ZU\hatX}(\hatx|x,y,z,u)}{Q_{XY|ZU}(x,y|z,u)}+\alpha\bigg(\bar{\mu}\bar{\beta}\log\frac{Q_{YZ|U}(y,z|u)}{P_{YZ}(y,z)}\\*
&\quad+\bar{\mu}\beta\log\frac{Q_Z(z)}{Q_{Z|U}(z|u)}+\mu d(x,\hatx)\bigg).\label{def:omegaamb}  
\end{align}
 }

Also define the negative cumulant generating functions 
\begin{align}
\nn&\Omega^{(\alpha,\mu,\beta,\theta)}(Q_{XYZU\hatX})\\*
&:=\!-\log \mathbb{E}_{Q_{XYZU\hatX}}\!\Big[\!\exp\big(\!-\!\theta\omega^{(\alpha,\mu,\beta)}_{Q_{XYZU\hatX}}\!(X,Y,Z,\hatX|U)\big)\Big],
\label{def:omegaamblambdaQP}\\*
&\Omega^{(\alpha,\mu,\beta,\theta)}
:=\min_{Q_{XYZU\hatX}\in\calQ}\Omega^{(\alpha,\mu,\beta,\theta)}(Q_{XYZU\hatX}).
\label{def:omegaamblambdaP}
\end{align}
\blue{
Finally, define the large-deviation rate functions 
\begin{align}
\nn&F^{(\alpha,\mu,\beta,\theta)}(R^{\rmi},R^{\rmc},D)\\*
&:=\frac{\Omega^{(\alpha,\mu,\beta,\theta)}-\theta\alpha\Big(\bar{\mu}(\bar{\beta} R^{\rmc}-R^{\rmi})+\mu D\Big)}{1+5\theta+\theta\alpha\bar{\mu}(3-\beta)},\label{def:lossyFalphamubetalambda}\\*
\nn&F(R^{\rmi},R^{\rmc},D)\\*
&:=\sup_{(\alpha,\theta,\mu,\beta)\in\bbR_+^2\times[0,1]^2}F^{(\alpha,\mu,\beta,\theta)}(R^{\rmi},R^{\rmc},D)\label{def:lossyF}.
\end{align}
}

\begin{lemma}
\label{propFOmega}
The following   hold.
\begin{itemize}
\item [i)] If $(R^{\rmi},R^{\rmc},D)\notin\calR$, then
\begin{align}
F(R^{\rmi},R^{\rmc},D)>0;
\end{align}
\item[ii)] If $(R^{\rmi},R^{\rmc},D)\in\calR$, then
\begin{align}
F(R^{\rmi},R^{\rmc},D)=0.
\end{align}
\end{itemize}
\end{lemma}
The proof of Lemma \ref{propFOmega} is similar to that of \cite[Property 4]{oohama2016wynerziv} and is given in Appendix \ref{proofpropFOmega}. We remark that Lemma \ref{propFOmega}, especially conclusion i), plays a central role in claiming the exponential strong converse theorem for the content identification problem with lossy recovery. As we will see shortly in Theorem \ref{mainresult:idlossy}, $F(R^{\rmi},R^{\rmc},D)$ in \eqref{def:lossyF} is a lower bound on the exponent of the probability of correct decoding.

\subsection{Exponential Strong Converse}
\label{sec:esconverse}
\begin{theorem}
\label{mainresult:idlossy}
For any encoding-decoding functions $(f^{(n)},g^{(n)})$ such that 
\begin{align}
\frac{1}{n}\log L&\leq R^{\rmc},~\frac{1}{n}\log M\geq R^{\rmi},
\end{align}
given any deterministic function $h^{(n)}$ and any distortion level $D$, we have
\begin{align}
\rmP_\rmc^{(n)}(f^{(n)},g^{(n)},h^{(n)},D)
&\leq 7\exp\big(-nF(R^{\rmi},R^{\rmc},D)\big)\label{expupppc}.
\end{align}
\end{theorem}
The proof of Theorem \ref{mainresult:idlossy} is given in Section \ref{proofmainresult:idlossy}. In the proof, we adapt the information spectrum method proposed by Oohama~\cite{oohama2016wynerziv,oohama2016new,oohama2015wak} to first establish a non-asymptotic upper bound on the probability of correct decoding. Invoking the upper bound (cf. Lemma \ref{fblidlossy2}) and applying Cram\'er's theorem on large deviations, we can further upper bound the probability of correct decoding. Subsequently, we proceed in a similar manner as \cite{oohama2016wynerziv,oohama2016new} to obtain the desired result. 

Second, we believe that both the image size characterization~\cite{csiszar2011information} and the perturbation approach~\cite{wei2009strong}   cannot lead to a strong converse theorem for the content identification problem with lossy recovery. The major difficulty lies in the fact that decoder needs to use the whole codebook $\calC=\{S(1),\ldots,S(M)\}$ and $Z^n$ to decode. Recall that $S(m)=f^{(n)}(Y^n(m))$ for $m\in\calM$. 

Invoking Lemma \ref{propFOmega} and Theorem \ref{mainresult:idlossy}, we conclude that the exponent in the right hand side of \eqref{expupppc} is strictly positive if the rate pairs are outside the rate-distortion region. Hence, we obtain the following exponential strong converse theorem.
\begin{theorem}
\label{strongconverse}
For any sequence of encoding-decoding-reproduction functions $(f^{(n)},g^{(n)},h^{(n)})$ such that 
\begin{align}
\limsup_{n\to\infty}\frac{1}{n}\log L&\leq R^{\rmc}\label{limsupl},~\liminf_{n\to\infty}\frac{1}{n}\log M\geq R^{\rmi},
\end{align}
given a distortion level $D$, we have that if $(R^{\rmi},R^{\rmc},D)\notin\calR$ (recall Theorem \ref{tuncelgidlossy}), then the probability of correct decoding vanishes to zero exponentially fast as $n$ goes to infinity.
\end{theorem}

Invoking Theorem \ref{strongconverse}, we conclude that the $\varepsilon$-rate distortion region satisfies $\calR(\varepsilon)=\calR^*$ for all $\varepsilon\in [0,1)$. Adopting the one-shot technique introduced in \cite{watanabe2015}, we can also establish a non-asymptotic achievability bound. \blue{Applying the Berry-Esseen theorem to the achievability bound and analyzing the bound in Lemma \ref{mainresult:idlossy}, we can conclude that the backoff from the boundary of the first-order region at finite blocklengths is of the order $\Theta(n^{-1/2})$.}

\subsection{Upper Bound on the Joint Identification-Error and Excess-distortion Exponent}
\begin{definition}
\label{def:eexponent}
A non-negative number $E$ is said to be an $(R^{\rmi},R^{\rmc},D)$-achievable joint identification-error and excess-distortion exponent if there exists a sequence of encoding-decoding-reproduction functions $(f^{(n)},g^{(n)},h^{(n)})$ such that \eqref{limsupl} holds and 
\begin{align}
\liminf_{n\to\infty}-\frac{\log \rmP_\rme^{(n)}(f^{(n)},g^{(n)},h^{(n)},D)}{n}\geq E.
\end{align}
The supremum of all $(R^{\rmi},R^{\rmc},D)$-achievable error exponent is called the optimal error exponent and denoted as $E^*(R^{\rmi},R^{\rmc},D)$.
\end{definition}

Recall that $\calR$ (Definition \ref{def:achratedistortion}) is the rate-distortion region with respect to $P_X,P_{Y|X},P_{Z|X}$. For any $Q_X,Q_{Y|X},Q_{Z|X}$, let $\calR(Q_X,Q_{Y|X},Q_{Z|X})$ be the rate-distortion region with respect to $Q_X,Q_{Y|X},Q_{Z|X}$. Invoking Lemma \ref{propFOmega}, Theorem \ref{mainresult:idlossy} and applying Marton's change-of-measure technique~\cite{Marton74}, we derive an upper bound on $E^*(R^{\rmi},R^{\rmc},D)$.

\begin{theorem}
\label{upjeedexponent}
The the optimal joint identification-error and excess-distortion exponent function satisfies
\begin{align}
\nn&E^*(R^{\rmi},R^{\rmc},D)\\*
&\leq \inf_{\substack{Q_{XYZ}:Z-X-Y\\(R^{\rmi},R^{\rmc},D)\notin\calR(Q_X,Q_{Y|X},Q_{Z|X})}} D(Q_{XYZ}\|P_{XYZ})\label{uppexpoent}.
\end{align}
\end{theorem}
Our main results for content identification with lossy recovery (Theorems \ref{strongconverse} and \ref{upjeedexponent}) can be specialized to the biometrical identification problem~\cite{willems2003capacity}, the content identification problem~\cite{tuncel2009capacity} and the Wyner-Ziv problem~\cite{wyner1976rate} since all these problems are special cases of the content identification problem with lossy recovery as argued in \cite{tuncel2014idenlossy}.

\vspace{-.05in}
\subsection{Extensions for the Biometrical Identification Problem}

In this subsection, we present several extensions for the biometrical identification problem~\cite{willems2003capacity}. The capacity (the maximum rate) of the biometrical identification problem was characterized by Willems {\em et al.} in~\cite{willems2003capacity}. Furthermore, the exponential strong converse theorem for the biometrical identification problem has been established in \cite[Theorem 2]{yagi2016}.  

Compared to the content identification with lossy recovery problem, there is no compression (no $f^{(n)}$) and no lossy recovery phase (no $h^{(n)}$) in the biometrical identification problem. Thus, $S(m)=Y^n(m)$ for each $m\in\calM$. Hence, the error probability is
\begin{align}
\rmP_{\rme}^{(n)}(g^{(n)})&:=\Pr\{\hatW\neq W\}.
\end{align}

Let $C_{\mathrm{bio}}$ be the capacity of the biometrical identification problem. Then, it can be verified that
\begin{align}
C_{\mathrm{bio}}
&=\sup \{R^{\rmi}:(R^{\rmi},\log|\calY|,d^+)\in\calR\}\\
&=I(P_Y,P_{Z|Y}).
\end{align}

Define the exponents
\begin{align}
 \underline{E}_{\rm{bio}}(R^\rmi)&:=\sup_{\lambda>0}\frac{\lambda R^\rmi\! -\! \log\mathbb{E}\left[\exp\left(\lambda \log\frac{P_{Z|Y}(Z|Y)}{P_Z(Z)}\right)\right] }{{1+ \lambda}},\\
 \overline{E}_{\rm{bio}}(R^\rmi)&:=\sup_{\lambda>0} \lambda R^\rmi-\log\mathbb{E}\left[\exp\left(\lambda \log\frac{P_{Z|Y}(Z|Y)}{P_Z(Z)}\right)\right].
\end{align}

\begin{theorem}
\label{bioext1}
For any decoding function $g^{(n)}$, we have that 
\begin{align}
\rmP_\rmc^{(n)}(g^{(n)})\leq 2\exp\big(-n\underline{E}_{\rm{bio}}(R^\rmi)\big). 
\end{align}
Furthermore, there exists a decoding function $g^{(n)}$ such that 
\begin{align}
\rmP_\rmc^{(n)}(g^{(n)})\geq \frac{1}{2}\exp\big(-n\overline{E}_{\rm{bio}}(R^\rmi)\big). 
\end{align}
\end{theorem}

It is easy to verify that $\underline{E}_{\rm{bio}}(R^\rmi)>0$ if $R^\rmi>C_{\rm{bio}}=I(P_Y,P_{Z|Y})$ and $\overline{E}_{\rm{bio}}=0$ if $R^\rmi\leq C_{\rm{bio}}$. Hence, the exponential strong converse theorem follows as a simple corollary. Although we cannot establish a  tight strong converse exponent, in the following, we present tight results on the moderate deviations constant. Let $\rmP_\rmc^*(n,R^\rmi)$ be the maximum probability of correct decoding when the number of items to be distinguished $M$  satisfies that $\log M\ge nR^\rmi$. Let
\begin{align}
\rmV:=\mathrm{Var}\left[\log\frac{P_{Z|Y}(Z|Y)}{P_Z(Z)}\right].
\end{align}
Throughout this section, we assume that $\rmV>0$.  Note that unlike the dispersion of a channel~\cite{hayashi2009information,polyanskiy2010finite}, $\rmV$ is the {\em unconditional information variance} instead of the optimized conditional information variance.

\begin{theorem}
\label{bioext2}
Consider any sequence of positive numbers $\{\xi_n\}_{n=1}^{\infty}$ such that $\xi_n\to 0$ and $\sqrt{n}\xi_n\to\infty$ as $n\to\infty$. When the rate $R^\rmi$ approaches capacity $C_{\rm{bio}}$ from above, the probability of correct decoding scales as
\begin{align}
\lim_{n\to\infty} -\frac{\log \rmP_\rmc^*(n,C_{\rm{bio}}+\xi_n)}{n\xi_n^2}=\frac{1}{2\rmV}\label{theorem:mdcsc}.
\end{align}
Similarly, when the rate $R^\rmi$ approaches capacity $C_{\rm{bio}}$ from below, the probability of correct decoding scales as
\begin{align}
\lim_{n\to\infty} -\frac{\log (1-\rmP_\rmc^*(n,C_{\rm{bio}}-\xi_n))}{n\xi_n^2}=\frac{1}{2\rmV}\label{theorem:mdc}.
\end{align}
\end{theorem}

The result in \eqref{theorem:mdcsc} implies that even if the rate $R^\rmi$ approaches the capacity from above with speed $\xi_n$, the probability of correct decoding still vanishes to zero (subexponentially fast). Similarly, \eqref{theorem:mdc} implies that if the rate $R^\rmi$ approaches the capacity from below with speed $\xi_n$, then the error probability vanishes to zero (subexponentially fast).

We remark that the study of moderate deviations for DMCs was done by Altu\u g and Wagner~\cite{altugwagner2014} and also by Polyanskiy and Verd\'u~\cite{polyanskiy2010channel}. For certain classes of quantum channels, moderate deviations analysis (above and below capacity) was done by Chubb, Tan, and Tomamichel~\cite{chubb2017moderate}. For other works on moderate deviations, see \cite{tan2012moderate,tan2014moderate,zhou2016second,altug2013lossless,zhou2016moderate,zhou2017jscc,zhou2017refined,zhou2017non}.

In the following, we also present the tradeoff between the number of items to be distinguished and the error probability when it is non-vanishing. Let $M^*(n,\varepsilon)$ be the maximum number of items to be distinguished such that the error probability satisfies $\min_{g^{(n)}}\rmP_\rme^{(n)}(g^{(n)})\leq \varepsilon$. The second-order coding rate for the biometrical identification problem is defined as 
\vspace{-0.05in}
\begin{align}
L^*(\varepsilon):=\liminf_{n\to\infty}\frac{1}{\sqrt{n}} \big(\log M^*(n,\varepsilon)-nI(P_Y,P_{Z|Y}) \big).
\end{align}
\begin{theorem}
\label{bioext3}
For any $\varepsilon\in(0,1)$, the second-order coding rate for the biometrical identification problem satisfies 
\begin{align}
L^*(\varepsilon)=\sqrt{\rmV}\Phi^{-1}(\varepsilon).
\end{align}
\end{theorem}
Theorem \ref{bioext3} implies that if we allow a non-vanishing error probability, then the rate $R^\rmi(n,\varepsilon):=\frac{1}{n}\log M^*(n,\varepsilon)$ approaches capacity $C_{\rm{bio}}$ with   speed $L^*(\varepsilon)/\sqrt{n}$. 

We remark that the study of second-order asymptotics dates back to Strassen~\cite{strassen1962asymptotische} and was revisted by Hayashi~\cite{hayashi2009information} and by Polyanskiy, Poor and Verd\'u~\cite{polyanskiy2010finite}. Also see \cite{TanBook}.

The proofs of Theorems \ref{bioext1}, \ref{bioext2} and \ref{bioext3} are given in Appendix~\ref{proofbioext}.

\section{Proof of Theorem \ref{mainresult:idlossy}}
\label{proofmainresult:idlossy}
\subsection{Preliminaries}
\label{sec:premnonasymp}
In this subsection, we present some definitions. Given an encoding function $f^{(n)}$ and any $m\in\calM$, let
\begin{align}
P_{S|Y^n}(s(m)|y^n(m)):=1\{s(m)=f^{(n)}(y^n(m))\}.
\end{align}
Given a deterministic function $h^{(n)}$ and any $w\in\calM$, let
\begin{align}
P_{\hatX^n|SZ^n}(\hatx^n|s(w),z^n):=1\{\hatx^n=h^{(n)}(s(w),z^n)\}.
\end{align}

For simplicity, we use $\bx$ to denote $x^n$, $\bx^M$ to denote $(x^n(1),\ldots,x^n(M))$, $s^M$ to denote $s(1),\ldots,s(M)$. In a similar manner, we have $\by,\bz,\mathbf{\hatx},\by^M$ and the corresponding random vectors $\bX,\bX^M,\bY,\bY^M,S^M,\bZ,\mathbf{\hatX}$. For simplicity, let $\bt:=(\bx^M,\by^M,s^M,z^n,\hatx^n)$, $\bT:=(\bX^M,\bY^M,S^M,\bZ,\mathbf{\hatX})$ and $\mathbf{\calT}=(\calX^{Mn},\calY^{Mn},\calL^M,\calZ^n,\hat{\calX}^n)$. Then let $P_{W\bT}$ be the joint distribution of $(W,\bX^M,\bY^M,S^M,\bZ,\mathbf{\hatX})$, induced by $P_W,P_X^n,P_{Y|X}^n,P_{Z|X}^n,P_{Y|S^n},P_{\hatX^n|SZ^n}$, i.e.,
\begin{align}
P_{W\bT}(w,\bt)
\nn&=\frac{1}{M}\Bigg(\prod_{m=1}^MP_X^n(x^n(m))P_{Y|X}^n(Y^n(m)|x^n(m))\\*
\nn&\qquad\times P_{S|Y^n}(s(m)|y^n(m))\Bigg)P_{Z|X}^n(z^n|x^n(w))\\*
&\qquad\times P_{\hatX^n|SZ^n}(\hatx^n|s(w),z^n)
\label{def:pwt}.
\end{align}

Furthermore, in this section, whenever we use $\mathbb{E}$, we mean the expectation over $P_{W\bT}$ unless otherwise stated. Note that for each $w\in\calM$, the joint distribution of $(\bX(w),\bY(w),S(w),\bZ,\mathbf{\hatX})$ is the same. Thus we let the joint distribution be $P_{\bX\bY S\bZ\mathbf{\bX}}$. In the following, all the distributions $P$ are induced by $P_{\bX\bY S\bZ\mathbf{\bX}}$.

Let $Q_{Y^n},Q_{X^n|SY^nZ^n},Q_{X^nY^n|SZ^n\hatX^n},Q_{Y^nZ^n|S},Q_{Z^n}$ be arbitrary distributions. Given $w\in\calM$ and any $\eta>0$, define the sets 
$\calA_i(w)$ for $i\in[1:7]$ as in \eqref{def:cala1} to \eqref{def:cala7} on the top of next page.
\begin{figure*}
\begin{align}
\calA_1(w)&:=\Big\{\bt:\frac{1}{n}\log\frac{P_{Y^n}(y^n(w))}{Q_{Y^n}(y^n(w))}\geq -\eta\Big\},\label{def:cala1}\\
\calA_2(w)&:=\Big\{\bt:\frac{1}{n}\log\frac{P_{Z^n|Y^n}(z^n|y^n(w))}{Q_{Z^n|SY^n}(z^n|s(w),y^n(w))}\geq -\eta\Big\},\label{def:cala2}\\
\calA_3(w)&:=\Big\{\bt:\frac{1}{n}\log\frac{P_{X^n|Y^nZ^n}(x^n(w)|y^n(w),z^n)}{Q_{X^n|SY^nZ^n}(y^n(w)|s(w),x^n(w),z^n)}\geq -\eta\Big\},\label{def:cala3}\\
\calA_4(w)&:=\Big\{\bt:\frac{1}{n}\log\frac{P_{X^nY^n|SZ^n}(x^n(w),y^n(w)|s(w),z^n)}{Q_{X^nY^n|SZ^n\hatX^n}(x^n(w),y^n(w)|s(w),z^n,\hatx^n)}\geq -\eta\Big\},\label{def:cala4}\\
\calA_5(w)&:=\Big\{\bt: R^{\rmc}\geq \frac{1}{n}\log \frac{Q_{Y^nZ^n|S}(y^n(w)z^n|s(w))}{P_{Y^nZ^n}(y^n(w),z^n)}-\eta\Big\},\label{def:cala5}\\
\calA_6(w)&:=\Big\{\bt:R^{\rmi}\leq \frac{1}{n}\log\frac{P_{Z^n|S}(z^n|s(w))}{Q_{Z^n}(z^n)}+\eta\Big\}\label{def:cala6},\\
\calA_7(w)&:=\Big\{\bt:d(x^n(w),\hatx^n)\leq D\Big\}\label{def:cala7}.
\end{align}
\hrulefill
\end{figure*}
Choose $U_i(W)=(S(W),Y^{i-1}(W),Z_{i+1}^n)$. Then it can be verified that $Z_i-X_i(W)-Y_i(W)-U_i(W)$ and $(X_i(W),Y_i(W))-(U_i(W),Z_i)-\hatX_i$ form two Markov chains under the joint distribution $P_{W\bT}$ (recall \eqref{def:pwt}). Furthermore, let $V_i(W)=(S(W),Z_{i+1}^n)$.

For $i=1,\ldots,n$, let $Q_{X_iY_iZ_iU_i\hatX_i}$ be any generic distributions and let $Q_{Y_i},Q_{Z_i},Q_{X_i|Y_iZ_iU_i},Q_{X_iY_i|Z_iU_i\hatX_i},Q_{Y_iZ_i|U_i}$ be induced by $Q_{X_iY_iZ_iU_i\hatX_i}$. Paralleling \eqref{def:cala1} to \eqref{def:cala6}, given any $\eta>0$, we define sets $\calB_i(w)$ for $i\in[1:7]$ as in \eqref{def:calb1} to \eqref{def:calb7} on the top of next page.
\begin{figure*}
\begin{align}
\calB_1(w)&:=\Bigg\{\bt:
0\geq \frac{1}{n}\sum_{i=1}^n\log\frac{Q_{Y_i}(y_i(w))}{P_{Y_i}(y_i(w))}-\eta\Bigg\},\label{def:calb1}\\
\calB_2(w)&:=\Bigg\{\bt:
0\geq \frac{1}{n}\sum_{i=1}^n \frac{Q_{Z_i|Y_iU_i}(x_i(w)|z_i,u_i(w))}{P_{Z_i|Y_i}(z_i|y_i(w))}-\eta\Bigg\},\label{def:calb2}\\
\calB_3(w)&:=\Bigg\{\bt:
0\geq \frac{1}{n}\sum_{i=1}^n \log \frac{Q_{X_i|Y_iZ_iU_i}(x_i(w)|y_i(w),z_i,u_i(w))}{P_{X_i|Y_iZ_i}(x_i(w)|y_i(w),z_i)}-\eta\Bigg\},\label{def:calb3}\\
\calB_4(w)&:=\Bigg\{\bt:
0\geq \frac{1}{n}\sum_{i=1}^n \log \frac{Q_{X_iY_i|Z_iU_i\hatX_i}(x_i(w),y_i(w)|z_i,u_i(w),\hatx_i)}{P_{X_iY_i|Z_iU_i}(x_i(w),y_i(w)|z_i,u_i(w))}-\eta\Bigg\},\label{def:calb4}\\
\calB_5(w)&:=\Bigg\{\bt:
R^{\rmc}-R^{\rmi}\geq \frac{1}{n}\sum_{i=1}^n\log \Bigg(\frac{Q_{Y_iZ_i|U_i}(y_i(w),z_i|u_i(w))}{P_{Y_iZ_i}(y_i(w),z_i)}\frac{Q_{Z_i}(z_i)}{P_{Z|V_i}(z_i|v_i(w))}\Bigg)-3\eta\Bigg\},\label{def:calb5}\\
\calB_6(w)&:=\Bigg\{\bt:
R^{\rmi}\leq \frac{1}{n}\sum_{i=1}^n \log\frac{P_{Z_i|V_i}(z_i|v_i(w))}{Q_{Z}(z_i)}+\eta\Bigg\},\label{def:calb6}\\
\calB_7(w)&:=\Bigg\{\bt:
D\geq \frac{1}{n}\sum_{i=1}^n \log e^{d(x_i(w),\hatx_i)}\Bigg\}.\label{def:calb7}
\end{align}
\hrulefill
\end{figure*}

\subsection{Proof of Theorem \ref{mainresult:idlossy}}
Invoking \eqref{def:pejointerrorexcessdis} and \eqref{def:pwt}, we define the fprobability of correct decoding as
\begin{align}
\nn&\rmP_\rmc^{(n)}(f^{(n)},g^{(n)},h^{(n)},D)\\*
&:=1-\rmP_{\rm{e}}^{(n)}(f^{(n)},g^{(n)},h^{(n)},D)\label{def:pcjointerrorexcessdis}\\
&=\sum_{w=1}^M\sum_{\substack{t\in\calT:\bz\in\calD(s^M,w)\\d(\bx(w),\mathbf{\hatx})\leq D}} P_{W\bT}(w,\bt)\label{def:expressionpcidlossy}.
\end{align}

We first present a non-asymptotic upper bound on $\rmP_\rmc^{(n)}(f^{(n)},g^{(n)},h^{(n)},D)$.
\begin{lemma}
\label{fblidlossy}
For any encoding-decoding functions $(f^{(n)},g^{(n)})$ such that 
\begin{align}
\frac{1}{n}\log L&\leq R^{\rmc}\label{lleqrrmc},~
\frac{1}{n}\log M\geq R^{\rmi},
\end{align}
given any deterministic function $h^{(n)}$ and any distortion level $D$, we have
\begin{align}
\nn&\rmP_\rmc^{(n)}(f^{(n)},g^{(n)},h^{(n)},D)\\*
&\leq P_{W\bT}\Bigg\{\bigcap_{i=1}^7\calA_i(W)\Bigg\}+6e^{-n\eta}.\label{upp1pcnid}
\end{align}
\end{lemma}
The proof of Lemma \ref{fblidlossy} is given in Appendix \ref{prooffblidlossy}. 

A few other remarks are in order. First, in the proof of Lemma \ref{fblidlossy}, we define seven sets for each $w\in\calM$ in \eqref{def:cala1} to~\eqref{def:cala7}. Equipped with these definitions, we obtain the upper bound in~\eqref{upp1pcnid} where the probability of correct decoding of $W$ depends on $S(W),X^n(W),Y^n(W),Z^n,\hatX^n$.

Second, in \eqref{upp1pcnid}, $\calA_5(w)$ corresponds to the enrollment phase, $\calA_6(w)$ corresponds to the identification phase and $\calA_7(w)$ corresponds to the lossy recovery phase. Furthermore, $\calA_1(w)$ to $\calA_4(w)$ are the auxiliary sets whose roles will be clear in subsequent analyses.

Third, the definitions of $\calQ$ (cf. \eqref{def:calQlossy}) and $\{\calA_i(w)\}_{i=1}^7$ are crucial. Note that $\calA_1(w)$ to $\calA_4(w)$ appear in Lemma \ref{fblidlossy}. They appear due to the different Markov conditions in the definitions of $\calP^*$ in \eqref{def:calPlossystar} and $\calQ$ in \eqref{def:calQlossy}. This is also closely related with the proof of Lemma \ref{propFOmega} in Appendix \ref{proofpropFOmega}. Hence, there is a subtle interplay between  Lemmas \ref{fblidlossy} and \ref{propFOmega}. This tension also appears in \cite{oohama2016wynerziv}. However, we need to adapt \cite{oohama2016wynerziv} to the content identification problem with lossy recovery carefully since these two problems are significantly different. 

Invoking Lemma \ref{fblidlossy} and choosing the distributions $Q_{Y^n}$, $Q_{X^n|SY^nZ^n}$, $Q_{X^nY^n|SZ^n\hatX^n}$, $Q_{Y^nZ^n|S}$, $Q_{Z^n}$ appropriately, we obtain the following lemma.
\begin{lemma}
Given the conditions in Lemma \ref{fblidlossy}, we have
\label{fblidlossy2}
\begin{align}
\nn&\rmP_{\rmc}^{(n)}(f^{(n)},g^{(n)},h^{(n)},D)\\*
&\leq P_{W\bT}\Bigg\{\bigcap_{i=1}^7 \calB_i(W)\Bigg\}+6e^{-n\eta}.\label{upp2pcnidlossy}
\end{align}
\end{lemma}
The proof of Lemma \ref{fblidlossy2} is given in Appendix \ref{prooffblidlossy2}.

A few remarks are in order. First, to introduce $U_i(W)$ in~\eqref{upp2pcnidlossy}, we make use of the Markov chain $(X_i(W),Y_i(W))-(S(W),Y^{i-1}(W),Z_i^n)-(X^{i-1}(W),Z^{i-1})$ which can be established similarly as \cite[Lemma 2]{oohama2016wynerziv}.

Second, note that in \cite{tuncel2014idenlossy}, Tuncel and G\"und\"uz chose the auxiliary random variable as $U_i(W)=(S(W),Z^{i-1},Z_{i+1}^n)$. If we choose $U_i(W)$ as in \cite{tuncel2014idenlossy}, we cannot obtain Lemma~\ref{fblidlossy2}. Furthermore, note that here we use both $U_i(W)$ and $V_i(W)$. This idea is also used in~\cite{oohama2016new}. In the proof of  Lemma \ref{lbunderlineOmegalossy}, we will eliminate $V_i(W)$ using H\"{o}lder's inequality.

\blue{Recall that $\bar{a}=1-a$ for $a\in[0,1]$.} In the following, for simplicity, we will use $Q_i$ to denote $Q_{X_iY_iZ_iU_i\hatX_i}$ and use $P_i$ to denote $P_{X_iY_iZ_iU_iV_i\hatX_i}$. Let $(\alpha,\lambda,\mu,\beta)\in\bbR_+\times [0,1]^2$. We need the following definitions to further upper bound the right hand side in Lemma \ref{fblidlossy2}. Let
\blue{
\begin{align}
\nn&f^{(\alpha,\mu,\beta)}_{Q_i,P_i}(x_i,y_i,z_i,\hatx_i|u_i,v_i)\\*
\nn&:=\frac{Q_{Y_i}(y_i)}{P_{Y_i}(y_i)}\frac{Q_{Z_i|Y_iU_i}(z_i|y_i,u_i)}{P_{Z_i|Y_i}(z_i|y_i)}
\frac{Q_{X_i|Y_iZ_iU_i}(x_i|y_i,z_i,u_i)}{P_{X_i|Y_iZ_i}(x_i|y_i,z_i)}
\\*
\nn&\qquad\times \frac{Q_{X_iY_i|Z_iU_i\hatX_i}(x_i,y_i|z_i,u_i(w),\hatx_i)}{P_{X_iY_i|Z_iU_i}(x_i,y_i|z_i,u_i)}\frac{Q_{Y_iZ_i|U_i}^{\alpha\bar{\mu}\bar{\beta}}(y_i,z_i|u_i)}{P_{Y_iZ_i}^{\alpha\bar{\mu}\bar{\beta}}(y_i,z_i)}\\*
&\qquad\times\frac{Q_{Z_i}^{\alpha\bar{\mu}\beta}(z_i)}{P_{Z_i|V_i}^{\alpha\bar{\mu}\beta}(z_i|v_i)} e^{\alpha \mu d(x_i,\hatx_i)}\label{def:fincramer}.
\end{align}
}

Next, we define the negative cumulant generating function of $f^{(\alpha,\mu,\beta)}_{Q_i,P_i}(\cdot)$ as in \eqref{def:omegaamblambda} on the next page.
\begin{figure*}
\begin{align}
\Omega^{(\alpha,\mu,\beta,\lambda)}(\{Q_i\}_{i=1}^n):=-\log\mathbb{E}\Big[\exp\Big(-\lambda\sum_{i=1}^n\log f^{(\alpha,\mu,\beta)}_{Q_i,P_i}(X_i(W),Y_i(W),Z_i,\hatX_i|U_i(W),V_i(W))\Big)\Big]\label{def:omegaamblambda}.
\end{align}
\hrulefill
\end{figure*}
For simplicity, also let
\begin{align}
\kappa(\alpha,\mu,\beta,\lambda)
&:=\lambda\alpha\bar{\mu}(\bar{\beta}R^{\rmc}-R^{\rmi}))+\lambda\alpha\mu D\label{def:kappa}.
\end{align}
Invoking Cram\'er's bound on large deviations~(cf. Lemma \ref{cramer}), we obtain the following lemma.
\begin{lemma}
\label{fblidlossy3}
For any $(\alpha,\lambda,\mu,\beta)\in\bbR_+^2\times[0,1]^2$ , given the conditions in Lemma 9, we have
\begin{align}
\nn&\rmP_\rmc^{(n)}(f^{(n)},g^{(n)},h^{(n)},D)\\*
&\!\!\!\leq 7\exp\left(-n\frac{\frac{1}{n}\Omega^{(\alpha,\mu,\beta,\lambda)}(\{Q_i\}_{i=1}^n)-\kappa(\alpha,\mu,\beta,\lambda)}{1+\lambda\big(4+\alpha\bar{\mu}(3-2\beta)\big)}\right).
\end{align}
\end{lemma}
The proof of Lemma \ref{fblidlossy3} is given in Appendix \ref{prooffblidlossy3}.

Let 
\begin{align}
\underline{\Omega}^{(\alpha,\mu,\beta,\lambda)}:=\inf_{n\geq 1}\sup_{\{Q_i\}_{i=1}^n}
\Omega^{(\alpha,\mu,\beta,\lambda)}
(\{Q_i\}_{i=1}^n).\label{def:lossyunderlineOmega}
\end{align}
Define
\begin{align}
\theta:=\frac{\lambda}{1-\lambda-\lambda\alpha\bar{\mu}\beta}\label{def:theta}.
\end{align}
Hence, we have
\begin{align}
\lambda=\frac{\theta}{1+\theta+\theta\alpha\bar{\mu}\beta}\label{lambdausetheta}.
\end{align}

The next lemma is essential in the proof.
\begin{lemma}
\label{lbunderlineOmegalossy}
For any $(\alpha,\lambda,\mu,\beta)\in\bbR_+^2\times[0,1]^2$ such that $\lambda\in(0,\frac{1}{1+\alpha\bar{\mu}\beta})$, we have that for $\theta$ in \eqref{def:theta},
\begin{align}
\underline{\Omega}^{(\alpha,\mu,\beta,\lambda)}\geq\frac{\Omega^{(\alpha,\mu,\beta,\theta)})}{1+\theta+\theta\alpha\bar{\mu}\beta}.
\end{align}
\end{lemma}
The proof of Lemma \ref{lbunderlineOmegalossy} is similar to that of  \cite[Proposition~2]{oohama2016wynerziv} and  given in Appendix \ref{prooflbunderlineOmegalossy}. In the proof of Lemma \ref{lbunderlineOmegalossy}, we first remove $W$ in the expression of $\Omega^{(\alpha,\mu,\beta,\lambda)}(\{Q_i\}_{i=1}^n)$. Subsequently, by adopting ideas from \cite{oohama2016wynerziv} and \cite{oohama2016new} and properly choosing distributions $Q_{X_iY_iZ_iU_i\hatX_i}$ via the recursive method, we can establish Lemma \ref{lbunderlineOmegalossy}.

Invoking Lemmas \ref{fblidlossy3} and \ref{lbunderlineOmegalossy}, we conclude that 
\begin{align}
&\rmP_\rmc^{(n)}(f^{(n)},g^{(n)},h^{(n)},D)\\
&\!\!\!\leq 7\exp\left(-n\frac{\underline{\Omega}^{(\alpha,\mu,\beta,\lambda)}-\lambda\alpha\bar{\mu}(\bar{\beta}R^{\rmc}-R^{\rmi})-\lambda\alpha\mu D}{1+\lambda\big(4+\alpha\bar{\mu}(3-2\beta)\big)}\right)\\
&\!\!\!\leq 7\exp\left(-n\frac{\frac{\Omega^{(\alpha,\mu,\beta,\theta)}}{1+\theta+\theta\alpha\bar{\mu}\beta}-\frac{\theta\alpha\bar{\mu}(\bar{\beta}R^{\rmc}-R^{\rmi})+\theta\alpha\mu D}{1+\theta+\theta\alpha\bar{\mu}\beta}}{1+\frac{\theta\big(4+\alpha\bar{\mu}(3-2\beta)\big)}{1+\theta+\theta\alpha\bar{\mu}\beta}}\right)\\
&\!\!\!=7\exp\Bigg(-n\frac{\Omega^{(\alpha,\mu,\beta,\theta)}-\theta\alpha\bar{\mu}(\bar{\beta}R^{\rmc}-R^{\rmi})-\theta\alpha\mu D}{1+5\theta+\theta\alpha\bar{\mu}(3-\beta)}\Bigg)\\
&\!\!\!=7\exp\big(-nF(R^{\rmi},R^{\rmc},D)\big),\label{useflossy}
\end{align}
where \eqref{useflossy} follows from the definition of $F(R^{\rmi},R^{\rmc},D)$ in~\eqref{def:lossyF}.
The proof of Theorem \ref{mainresult:idlossy} is now complete.

\section{Conclusion}
\label{sec:conclusion}
In this paper, we  derived  a non-asymptotic converse bound for content identification problem with lossy recovery. Invoking the non-asymptotic bound, we established an exponential strong converse theorem. As a corollary of our main result, we derived an upper bound on the optimal exponent of the joint identification-error and excess-distortion probability. Our main results can be specialized to the biometrical identification problem~\cite{willems2003capacity} and the content identification problem~\cite{tuncel2009capacity}.

There are several avenues for future research. First, note that in Theorem \ref{mainresult:idlossy}, we present only a non-asymptotic exponential type upper bound on the probability of correct decoding.  Although this is sufficient for us to claim the exponential strong converse theorem~(cf. Theorem \ref{strongconverse}) by invoking Lemma~\ref{propFOmega}, it is worth deriving the exact exponent for the probability of correct decoding. The ideas involved in characterizing the exact strong converse exponent in~\cite{oohama1994universal,kang2015} and the one-shot techniques in \cite{watanabe2015,yassaee2013technique} might be useful. Second, after Theorem \ref{strongconverse}, we remarked that the second-order coding terms are in the order of $\Theta(n^{-1/2})$. In the future, one may be interested in nailing down the exact second-order coding region. For this line of research, one may borrow ideas from \cite{scarlett2015, watanabe2015second, zhou2015second}. Third, in this paper, we only considered the discrete memory sources and discrete memoryless channels. In future, one may consider Gaussian memoryless sources, the additive Gaussian white noise channel, and the quadratic distortion measure. For the special case of biometrical identification problem, the capacity for Gaussian case was derived in~\cite{willems2003capacity}. However, for Gaussian case of content identification with lossy recovery, one has to first calculate the rate-distortion region~(cf.\ Theorem~\ref{tuncelgidlossy}). To do so, it is necessary to check whether Gaussian test channels are first-order optimal by referring to~\cite{geng2014gau} and \cite{liu2016brascamp}. The strong converse theorem for Gaussian case  may be inspired by works of Fong and Tan in \cite{fong2016mac} and \cite{fong2015proof}.

\appendix
\subsection{Proof of Lemma \ref{fblidlossy}}
\label{prooffblidlossy}

Recall that $\bx=x^n$, $\by=y^n$, $\bz=z^n$, $\bt=(\bx^M,\by^M,s^M,\bz,\mathbf{\hatx})$ and we will drop the subscript of distributions when there is no confusion. Recall the definition of the distribution $P_{W\bT}$ in \eqref{def:pwt} and the definitions of $\{\calA_i(w)\}_{i=1}^7$ in \eqref{def:cala1} to \eqref{def:cala6}. Invoking \eqref{def:expressionpcidlossy} and noting that $d(\bx(w),\mathbf{\hatx})\leq D$ is equivalent to $t\in\calA_7(w)$, we have
\begin{align}
\nn&\rmP_\rmc^{(n)}(f^{(n)},g^{(n)},h^{(n)},D)\\*
\nn&=\sum_{w=1}^M\Bigg(\sum_{\substack{\bt\in\big(\bigcap_{i=1}^6\calA_i(w)\big):\\\bt\in\calA_7(w),~z^n\in\calD(s^M,w)}}P_{W\bT}(w,\bt)\\*
&\qquad\qquad\quad+\sum_{\substack{\bt\in\big(\bigcup_{i=1}^6\calA_i^{\rmc}(w)\big)\big):\\\bt\in\calA_7(w),~z^n\in\calD(s^M,w)}}P_{W\bT}(w,\bt)\Bigg)\label{pcnidlossy}.
\end{align}

Ignoring the constraint that $z^n\in\calD(s^M,w)$, we can upper bound the first term in \eqref{pcnidlossy} by
\begin{align}
\Delta_1
&=\sum_{w=1}^M\frac{1}{M}\sum_{\substack{\bt\in\bigcap_{i=1}^7\calA_i(w)}}\!\!\!\!P(\bt)=P_{W\bT}\Bigg\{\bigcap_{i=1}^7 \calA_i(W)\Bigg\}\label{uppDelta1pcnidlossy}.
\end{align}
For $i=2,\ldots,7$, let
\begin{align}
\Delta_i:=\sum_{w=1}^M\sum_{\substack{\bt\in\calA_{i-1}^{\rmc}(w)\bigcap\calA_7(w):\\z^n\in\calD(s^M,w)}}P_{W\bT}(w,\bt).
\end{align}
Then, the second term in \eqref{pcnidlossy} is no larger than $\sum_{i=2}^7 \Delta_i$ by the union bound.

For simplicity, let $\bt(w):=(\bx(w),\by(w),s(w),\bz,\mathbf{\hatx})$. Invoking \eqref{def:cala1}, in a similar manner as the proof of \cite[Lemma 12]{oohama2016wynerziv}, we obtain that
\begin{align}
\Delta_2
&=\sum_{w=1}^M\sum_{\substack{\bt\in\calA_1^{\rmc}(w)\bigcap\calA_7(w):\\\bz\in\calD(s^M,w)}}P(\bt)\\
&\leq \sum_{w=1}^M\sum_{\substack{\bt\in\calA_1^{\rmc}(w)}}P(\bt)\label{ignorecons}\\
&=\sum_{w=1}^M\frac{1}{M}\sum_{\substack{\bt(w):\\ P(\by(w))\leq e^{-n\eta}Q(\by(w))}}P(\bt(w))\label{usecalb1}\\
&\leq \sum_{w=1}^M\frac{1}{M}\sum_{\substack{\by(w): P(\by(w))\leq e^{-n\eta}Q(\by(w))}}P(\by(w))\label{sumothers}\\
&\leq e^{-n\eta}\sum_{w=1}^M\frac{1}{M}\sum_{\by(w)}Q(\by(w))\\*
&\leq e^{-n\eta},\label{uppDelta2pcnidlossy}
\end{align}
where \eqref{ignorecons} follows from ignoring the constraints that $\big(\bx^M,\by^M,s^M,\bz,\mathbf{\hatx}\big)\in\calA_7(w)$ and $z^n\in\calD(s^M,w)$; \eqref{usecalb1} follows from the definition of $\calA_1(w)$ in \eqref{def:cala1} and the fact that for each $w$,
\begin{align}
\sum_{m\neq w,m\in\calM}\sum_{\bx(m),\by(m),s(m)}P(\bt)
&=P(\bt(w))\label{removenotw},
\end{align}
and \eqref{sumothers} follows since $\sum_{\bx(w),s(w),\bz,\mathbf{\hatx}}P(\bt(w))=P(\by(w))$.

Similarly as \eqref{uppDelta2pcnidlossy}, using \eqref{def:cala2} and \eqref{def:cala3}, we obtain
\begin{align}
\Delta_3&\leq e^{-n\eta},\label{uppDelta3pcnidlossy}\\
\Delta_4&\leq e^{-n\eta}.\label{uppDelta4pcnidlossy}
\end{align}

Invoking the definition of $\calA_4(w)$ in~\eqref{def:cala4}, we conclude that
\begin{align}
\Delta_5
&\leq \sum_{w=1}^M \sum_{\substack{\bt\in\calA_4^{\rmc}(w)}}P(\bt)\label{ignoreconsDelta5}\\
&=\sum_{w=1}^M\frac{1}{M}\sum_{\substack{\bt(w):P(\bx(w),\by(w)|s(w),\bz)\leq\\ e^{-n\eta}Q(\bx(w),\by(w)|s(w),\bz,\mathbf{\hatx})}} P(\bt(w))\\
\nn&=\sum_{w=1}^M\frac{1}{M}\sum_{\substack{\bt(w):P(\bx(w),\by(w)|s(w),\bz)\leq \\e^{-n\eta}Q(\bx(w),\by(w)|s(w),\bz,\mathbf{\hatx})}} 
\Big(P(s(w),\bz)\\
&\qquad\qquad\times P(\bx(w),\by(w)|s(w),\bz)P(\mathbf{\hatx}|s(w),\bz)\Big)\label{use2markov}\\
\nn&\leq \sum_{w=1}^M\frac{1}{M}\sum_{\substack{\bt(w)}}
\Big(P(s(w),\bz)P(\mathbf{\hatx}|s(w),\bz)\\*
&\qquad\qquad\qquad\times e^{-n\eta}Q(\bx(w),\by(w)|s(w),\bz,\mathbf{\hatx})\Big)\\*
&\leq e^{-n\eta}\label{uppDelta5pcnidlossy},
\end{align}
where \eqref{ignoreconsDelta5} follows similarly as \eqref{ignorecons} and \eqref{use2markov} follows due to the Markov chain $(X^n(W),Y^n(W))-(S(W),Z^n)-\hatX^n$.

Then, invoking \eqref{def:cala5}, we upper bound $\Delta_6$ as follows:
\begin{align}
\Delta_6
&=\sum_{w=1}^M \sum_{\substack{\bt\in\calA_5^{\rmc}(w)\bigcap\calA_7(w):\\z^n\in\calD(s^M,w)}}P(\bt)\\
&\leq \sum_{w=1}^M \frac{1}{M}
\sum_{\substack{\bt(w):P(\by(w),z^n)\leq \\ Q(\by(w),z^n|s(w))e^{-n(\eta+R^{\rmc})}}}
P(\bt(w))\label{ignoreconsuseb4}\\
\nn&=\sum_{w=1}^M \frac{1}{M}\sum_{\substack{(\by(w),s(w),\bz):P(\by(w),\bz)\\\leq Q(\by(w),z^n|s(w))e^{-n(\eta+R^{\rmc})}}}\Big(P(\by(w),z^n)\\*
&\qquad\qquad\quad\times P(s(w)|\by(w))\sum_{\mathbf{\hatx}}P(\mathbf{\hatx}|s(w),\bz)\Big)\\
&\leq \sum_{w=1}^M \frac{1}{M}\sum_{\by(w),s(w),\bz}Q(\by(w),z^n|s(w))e^{-n(\eta+R^{\rmc})}\label{psgywleq1}\\*
&=\sum_{w=1}^M \frac{1}{M}\sum_{s(w)}e^{-n(\eta+R^{\rmc})}\\
&\leq e^{-n\eta}\label{uppDelta6pcnidlossy},
\end{align}
where \eqref{ignoreconsuseb4} follows from dropping   $\big(\bx^M,\by^M,s^M,\bz,\mathbf{\hatx}\big)\in\calA_7(w)$, invoking the definition of $\calA_4(w)$ in \eqref{def:cala4} and using~\eqref{removenotw}; \eqref{psgywleq1} follows since $P(s(w)|\by(w))\leq 1$ for all $w\in\calM$; and \eqref{uppDelta6pcnidlossy} follows since $\sum_{s(w)}=|\calL|=L$ for each $w\in\calM$ and the fact that $L\leq e^{nR^{\rmc}}$ from \eqref{lleqrrmc}.

Finally, invoking \eqref{def:cala6}, we upper bound $\Delta_7$ as follows:
\begin{align}
\Delta_7
&\leq \sum_{w=1}^M\sum_{\substack{\bt\in\calA_6^{\rmc}(w):\\z^n\in\calD(s^M,w)}}P(\bt)\label{ignorecons6}\\
&=\sum_{w=1}^M\sum_{\substack{\bt:
\bz\in\calD(s^M,w)\\P(\bz|s(w))\leq e^{nR^{\rmi}}e^{-n\eta}Q(\bz)}}
P(\bt)\label{dropdisuseb6}\\
&=\sum_{w=1}^{M}\frac{1}{M}\sum_{\substack{s^M,\bz:\bz\in\calD(s^M,w)\\P(\bz|s(w))\leq e^{nR^{\rmi}}e^{-n\eta}Q(\bz)}}\!\!\!\!P(s^M)P(\bz|s(w))\label{sumtooneandrearrange}\\
&\leq \sum_{w=1}^{M}\frac{1}{M}\!\!\!\!\sum_{\substack{s^M,\bz:\bz\in\calD(s^M,w)\\P(\bz|s(w))\leq e^{nR^{\rmi}}e^{-n\eta}Q(\bz)}}\!\!\!\!\!\!\!\!\!\!\!\!P(s^M)e^{nR^{\rmi}}e^{-n\eta}Q(\bz)\\
&\leq e^{-n\eta}\frac{e^{nR^{\rmi}}}{M}\sum_{w=1}^{M}\sum_{s^M}\sum_{z^n:z^n\in\calD(s^M,w)}P(s^M)Q(z^n)\label{usemgeq}\\
&\leq e^{-n\eta}\sum_{w=1}^{M}\sum_{s^M}P(s^M)Q(\calD(s^M,w))\\
&=e^{-n\eta}\sum_{s^M}P(s^M)Q\Big(\bigcup_{w=1}^M\calD(s^M,w)\Big)\label{usedisjointd}.\\
&\leq e^{-n\eta}\label{uppDelta7pcnidlossy}
\end{align}
where \eqref{ignorecons6} follows from dropping  
$\big(\bx^M,\by^M,s^M,\bz,\mathbf{\hatx}\big)\in\calA_7(w)$; \eqref{dropdisuseb6} follows invoking the definition of $\calA_6(w)$ in \eqref{def:cala6}; \eqref{sumtooneandrearrange} follows since for each $w\in\calM$,
\begin{align}
\prod_{\substack{i\neq w\\ i\in\calM}} \sum_{\bx(i),\by(i)}P(\bx(i),\by(i),s(i))&=\prod_{\substack{i\neq w\\ i\in\calM}}P_S(s(i)),
\end{align}and
$\sum_{\bx(w),\by(w)} P(x(w),y(w),s(w))P(\bz|\bx(w))
 =P(s(w),\bz)=P(s(w))P(\bz|s(w))$; \eqref{usemgeq} follows since $M\geq e^{nR^{\rmi}}$ due to \eqref{lleqrrmc}; \eqref{usedisjointd} follows since that decoding regions are disjoint for different $w\in\calM$.

The proof of Lemma \ref{fblidlossy} is complete by combining \eqref{uppDelta1pcnidlossy}, \eqref{uppDelta2pcnidlossy}, \eqref{uppDelta3pcnidlossy}, \eqref{uppDelta4pcnidlossy}, \eqref{uppDelta5pcnidlossy}, \eqref{uppDelta6pcnidlossy} and \eqref{uppDelta7pcnidlossy}.

\vspace{-0.05in}
\subsection{Proof of Lemma \ref{fblidlossy2}}
\label{prooffblidlossy2}

Recall that in Section \ref{sec:premnonasymp}, we set $U_i(W)=(S(W),Y^{i-1}(W),Z_{i+1}^n)$ and  $V_i(W)=(S(W),Z_{i+1}^n)$. Then, in the following, let $U_i=(S,Y^{i-1},Z_{i+1}^n)$ and $V_i=(S,Z_{i+1}^n)$. Recall that for $i=1,\ldots,n$,  $Q_{X_iY_iZ_iU_i\hatX_i}$ is any generic distribution and $Q_{Y_i},Q_{Z_i},Q_{X_i|Y_iZ_iU_i},Q_{X_iY_i|Z_iU_i\hatX_i},Q_{Y_iZ_i|U_i}$ are induced by $Q_{X_iY_iZ_iU_i\hatX_i}$. Furthermore, note that in Lemma \ref{fblidlossy}, we are free to choose the distributions $Q_{Y^n}$, $Q_{Z^n}$, $Q_{Z^n|SY^n}$, $Q_{X^n|SY^nZ^n}$, $Q_{X^nY^n|SZ^n\hatX^n}$, $Q_{Y^nZ^n|S}$. Our choices for these distributions are as follows:
\begin{align}
&Q_{Y^n}(y^n):=\prod_{i=1}^n Q_{Y_i}(y_i)\label{lossychoseQY},~Q_{Z^n}(z^n):=\prod_{i=1}^n Q_{Z_i}(z_i)\\
\nn&Q_{Z^n|SY^n}(z^n|s,y^n)\\*
&:=\prod_{i=1}^n Q_{Z_i|SY^iZ_{i+1}^n}(z_i|s,y^i,z_{i+1}^n)\\*
&=\prod_{i=1}^n Q_{Z_i|Y_iU_i}(z_i|y_i,u_i)\\
\nn&Q_{X^n|SY^nZ^n}(x^n|s,y^n,z^n)\\*
&:=\prod_{i=1}^n Q_{X_i|SY^iZ_i^n}(x_i|s,y^i,z_i^n)\\
&=\prod_{i=1}^n Q_{X_i|Y_iZ_iU_i}(x_i|y_i,z_i,u_i)\\
\nn&Q_{X^nY^n|SZ^n\hatX^n}(x^n,y^n|s,z^n,\hatx^n)\\*
&:=\prod_{i=1}^n Q_{X_iY_i|SY^{i-1}Z_i^n\hatX}(x_i,y_i|s,y^{i-1},z_i^n,\hatx_i)\\*
&=\prod_{i=1}^n Q_{X_iY_i|Z_iU_i\hatX_i}(x_i,y_i|z_i,u_i,\hatx_i),\\
\nn&Q_{Y^nZ^n|S}(y^n,z^n|s)\\*
&:=\prod_{i=1}^n Q_{YZ|SY^{i-1}Z_{i+1}^n}(y_i,z_i|s,y^{i-1},z_{i+1}^n)\\*
&=\prod_{i=1}^n Q_{YZ|U_i}(y_i,z_i|u_i)\label{lossychoseQYZgS}.
\end{align}

Recall from Section \ref{sec:premnonasymp} that for each $w\in\calM$, the joint distribution of $(X^n(W),Y^n(W),S(W),Z^n,\hatX^n)$ is the same and denoted as $P_{X^nY^nSZ^n\hatX^n}$. The marginal distributions of $P_{X^nY^nSZ^n\hatX^n}$ are as follows:
\begin{align}
&P_{Y^n}(y^n)=\prod_{i=1}^n P_{Y_i}(y_i),\label{lossypy}~P_{Z^n}(z^n)=\prod_{i=1}^n P_{Z_i}(z_i),\\
&P_{Z^n|Y^n}(z^n|y^n)=\prod_{i=1}^n P_{Z_i|Y_i}(z_i|y_i),\\
&P_{X^n|Y^nZ^n}(x^n|y^n,z^n)=\prod_{i=1}^n P_{X_i|Y_iZ_i}(x_i|y_i,z_i),\\
\nn&P_{X^nY^n|SZ^n}(x^n,y^n|s,z^n)\\*
&=\prod_{i=1}^n P_{X_iY_i|SX^{i-1}Y^{i-1}Z^n}(x_i,y_i|s,x^{i-1},y^{i-1},z^n)\\*
&=\prod_{i=1}^n P_{X_iY_i|SY^{i-1}Z_i^n}(x_i,y_i|s,y^{i-1},z_i^n)\label{checkmarkov}\\*
&=\prod_{i=1}^n P_{X_iY_i|Z_iU_i}(x_i,y_i|z_i,u_i)\label{lossypxygsz},\\
&P_{Y^nZ^n}(y^n,z^n)=\prod_{i=1}^nP_{Y_iZ_i}(y_i,z_i),\label{pyzn}\\
&P_{Z^n|S}(z^n|s)=\prod_{i=1}^n P_{Z_i|SZ_{i+1}^n}(z_i|s,z_{i+1}^n)\\*
&\qquad\qquad\quad=\prod_{i=1}^n P_{Z_i|V_i}(z_i|v_i)\label{lossypzgs},
\end{align}
where \eqref{checkmarkov} holds since the Markov chain $(X_i(W),Y_i(W))-(S(W),Y^{i-1}(W),Z_i^n)-(X^{i-1}(W),Z^{i-1})$ holds. The proof of this Markov chain is similar as \cite[Lemma 2]{oohama2016wynerziv} and thus omitted.

Recall the definitions of $\{\calB_i(w)\}_{i=1}^7$ in \eqref{def:calb1} to \eqref{def:calb7}. For each $w\in\calM$, let
\begin{align}
\nn\tilde{\calB}_5(w)
&:=\Bigg\{\bt:
R^{\rmc}+\eta\geq \\*
&\qquad\frac{1}{n}\sum_{i=1}^n\log\frac{Q_{Y_iZ_i|U_i}(y_i(w),z_i|u_i(w))}{P_{Y_iZ_i}(y_i(w),z_i)}\Bigg\}.
\end{align}
We remark that $\tilde{\calB}_5(w)$ corresponds to $\calA_5(w)$ (recall \eqref{def:cala5}) in Lemma \ref{fblidlossy} by applying the choice of $Q_{Y^nZ^n|S}$ in \eqref{lossychoseQYZgS} and the definition in \eqref{pyzn}.

Recall the definition of $P_{W\bT}$ in Section \ref{sec:premnonasymp}. Using Lemma~\ref{fblidlossy} and   \eqref{lossychoseQY}--\eqref{lossychoseQYZgS} and \eqref{lossypy}--\eqref{lossypzgs},  we obtain
\begin{align}
\nn&\rmP_{\rmc}^{(n)}(f^{(n)},g^{(n)})\\*
&\leq P_{W\bT}\Bigg\{\bigcap_{i=1,i\neq 5}^7\calB_i(W)\bigcap\tilde{\calB}_5(W)\Bigg\}+6e^{-n\eta}\label{lemmapcnidlossy}.
\end{align}

For each $w\in\calM$, when $t\in\bigcap_{i=1,i\neq 5}^7\calB_i(w)\bigcap\tilde{\calB}_5(w)$, invoking the constraints related with $\tilde{\calB}_5(w)$ and $\calB_6(w)$, we have that
\begin{align}
R^{\rmc}-R^{\rmi}
\nn&\geq \frac{1}{n}\sum_{i=1}^n\Bigg(\log\frac{Q_{Y_iZ_i|U_i}(y_i(w),z_i|u_i(w))}{P_{Y_iZ_i}(y_i(w),z_i)}\\*
&\qquad\qquad\qquad-\log\frac{P_{Z_i|V_i}(z_i|v_i(w)}{Q_{Z_i}(z_i)}\Bigg)-3\eta\\
\nn&=\frac{1}{n}\sum_{i=1}^n\Bigg(\log\frac{Q_{Y_iZ_i|U_i}(y_i(w),z_i|u_i(w))}{P_{Y_iZ_i}(y_i(w),z_i)}\\*
&\qquad\qquad\qquad\quad\times\frac{Q_{Z_i}(z_i)}{P_{Z_i|V_i}(z_i|v_i(w)}\Bigg)-3\eta\label{rc-ri}.
\end{align}
Hence, for each $w\in\calM$, when $t\in\bigcap_{i=1,i\neq 5}^7\calB_i(w)\bigcap\tilde{\calB}_5(w)$, we have $t\in\bigcap_{i=1,i}^7\calB_i(w)$ (recall \eqref{def:calb5}). Thus,
\begin{align}
\nn&P_{W\bT}\Bigg\{\bigcap_{i=1,i\neq 5}^7\calB_i(W)\bigcap\tilde{\calB}_5(W)\Bigg\}\\*
&\leq P_{W\bT}\Bigg\{\bigcap_{i=1}^7\calB_i(W)\Bigg\}.
\end{align}

The proof of Lemma \ref{fblidlossy2} is now complete.

\subsection{Proof of Lemma \ref{fblidlossy3}}
\label{prooffblidlossy3}
\blue{Recall that $\bar{a}=1-a$ for $a\in[0,1]$.} For each $w\in\calM$ and any $(\alpha,\lambda,\mu,\beta)\in\bbR_+^2\times [0,1]^2$, define $\calF_i(w)=\calB_i(w)$ (cf.~\eqref{def:calb1} to \eqref{def:calb4}) for $i\in[1:4]$ and define the sets $\calF_i(w)$ for $i\in[5:7]$ as in \eqref{def:calf5} to \eqref{def:calf7} on the top of next page.
\begin{figure*}
\begin{align}
\calF_5(w)&:=\Bigg\{\bt:
\alpha\bar{\mu}\bar{\beta}\Big(R^{\rmc}-R^{\rmi}\Big)\geq  \alpha\bar{\mu}\bar{\beta}\frac{1}{n}\sum_{i=1}^n\log \Bigg(\frac{Q_{Y_iZ_i|U_i}(y_i(w),z_i|u_i(w))}{P_{Y_iZ_i}(y_i(w),z_i)}\frac{Q_{Z_i}(z_i)}{P_{Z|V_i}(z_i|v_i(w)}\Bigg)-3\alpha\bar{\mu}\bar{\beta}\eta\Bigg\},\label{def:calf5}\\
\calF_6(w)&:=\Bigg\{\bt:
\alpha\bar{\mu}\beta R^{\rmi}\leq \frac{\alpha\bar{\mu}\beta}{n}\sum_{i=1}^n \log\frac{P_{Z_i|V_i}(z_i|v_i(w))}{Q_{Z}(z_i)}+\alpha\bar{\mu}\beta\eta\Bigg\},\label{def:calf6}\\
\calF_7(w)&:=\Bigg\{\bt:
\alpha\mu D \geq \frac{\alpha\mu}{n}\sum_{i=1}^n \log e^{d(x_i(w),\hatx_i)}\Bigg\}.\label{def:calf7}
\end{align}
\hrulefill
\end{figure*}

Invoking Lemma \ref{fblidlossy2}, for any $(\alpha,\lambda,\mu,\beta)\in\bbR_+^2\times [0,1]^2$, we obtain
\begin{align}
\nn&\rmP_\rmc^{(n)}(f^{(n)},g^{(n)},h^{(n)},D)\\*
&\leq P_{W\bT}\Bigg\{\bigcap_{i=1}^7\calF_i(W)\Bigg\}+6e^{-n\eta}\label{lbprcalf}.
\end{align}

We make use of the following Cram\'er's bound for large deviations.
\begin{lemma}
\label{cramer}
For any real valued random variable $Z$ and any $\lambda>0$, we have
\begin{align}
\Pr\{Z\geq a\}\leq \exp\big(-(\lambda a-\log\mathbb{E}[\exp(\lambda Z)]\big).
\end{align}
\end{lemma}

Let 
\begin{align}
R(\alpha,\mu,\beta)
&:=
\alpha\bar{\mu}(\bar{\beta}R^{\rmc}-R^{\rmi})+\alpha\mu D\label{def:rincramer},\\
c(\alpha,\mu,\beta)
&:=4+3\alpha\bar{\mu}\bar{\beta}+\alpha\bar{\mu}\beta\\
&=4+\alpha\bar{\mu}(3-2\beta)\label{def:cambeta}.
\end{align}

Recall the definition of $P_{W\bT}$ in Section \ref{sec:premnonasymp}. In this subsection, whenever we use $\Pr$, we mean the probability with respect to $P_{W\bT}$ unless otherwise stated.
Recall the definition of $f^{(\alpha,\mu,\beta)}_{Q_i,P_i}(\cdot)$ in \eqref{def:fincramer}.
Combining \eqref{lbprcalf}, \eqref{def:rincramer}, \eqref{def:fincramer} and Lemma \ref{cramer}, we obtain that for any $\lambda\in\bbR_+$, the probability of correct decoding can be upper bounded as in \eqref{crameridlossy2} on the top of next page,
\begin{figure*}
\begin{align}
\nn&\rmP_\rmc^{(n)}(f^{(n)},g^{(n)},h^{(n)},D)\\*
&\leq \Pr\Big\{n\Big(R(\alpha,\mu,\beta)+c(\alpha,\mu,\beta)\eta\Big)\geq \sum_{i=1}^n\log f^{(\alpha,\mu,\beta)}_{Q_i,P_i}(X_i(W),Y_i(W),Z_i,\hatX_i|U_i(W),V_i(W))\Big\}+6e^{-n\eta}\\
&=\Pr\Big\{-\sum_{i=1}^n\log f^{(\alpha,\mu,\beta)}_{Q_i,P_i}(X_i(W),Y_i(W),Z_i,\hatX_i|U_i(W),V_i(W))\geq -n\Big(R(\alpha,\mu,\beta)+c(\alpha,\mu,\beta)\eta\Big)\Big\}+6e^{-n\eta}\\
&\leq \exp\Big\{n\lambda \Big(R(\alpha,\mu,\beta)+c(\alpha,\mu,\beta)\eta\Big)\nn\\
&\qquad +\log\mathbb{E}\Big[\exp\Big(-\lambda\sum_{i=1}^n\log f^{(\alpha,\mu,\beta)}_{Q_i,P_i}(X_i(W),Y_i(W),Z_i,\hatX_i|U_i(W),V_i(W))\Big)\Big]\Big\}+6e^{-n\eta}\\
&=\exp\Big\{n\Big(\lambda R(\alpha,\mu,\beta)+\lambda c(\alpha,\mu,\beta)\eta-\frac{1}{n}\Omega^{(\alpha,\mu,\beta,\lambda)}(\{Q_i\}_{i=1}^n)\Big)\Big\}+6e^{-n\eta}\label{crameridlossy2},
\end{align}
\hrulefill
\end{figure*}
where \eqref{crameridlossy2} follows from \eqref{def:omegaamblambda}.

Choose $\eta$ such that
\begin{align}
-\eta
\nn&=\lambda R(\alpha,\mu,\beta)+\lambda c(\alpha,\mu,\beta)\eta\\*
&\qquad-\frac{1}{n}\Omega^{(\alpha,\mu,\beta,\lambda)}(\{Q_i\}_{i=1}^n).
\end{align}
Thus,
\begin{align}
\eta=\frac{\frac{1}{n}\Omega^{(\alpha,\mu,\beta,\lambda)}(\{Q_i\}_{i=1}^n)-\lambda R(\alpha,\mu,\beta)}{1+\lambda c(\alpha,\mu,\beta)}\label{choseetaidlossy}.
\end{align}
Using the bound in \eqref{crameridlossy2}, the definition in~\eqref{choseetaidlossy} and recalling the definitions of $\kappa(\cdot)$ in \eqref{def:kappa}, $R(\cdot)$ in \eqref{def:rincramer} and $c(\cdot)$ in \eqref{def:cambeta}, we obtain that
\begin{align}
\nn&\rmP_\rmc^{(n)}(f^{(n)},g^{(n)},h^{(n)},D)\leq 7\exp(-n\eta)\\*
&\!\!\leq 7\exp\left(-n\frac{\frac{1}{n}\Omega^{(\alpha,\mu,\beta,\lambda)}(\{Q_i\}_{i=1}^n)-\kappa(\alpha,\mu,\beta,\lambda)}{1+\lambda\big(4+\alpha\bar{\mu}(3-2\beta)\big)}\right).
\end{align}
The proof of Lemma \ref{fblidlossy3} is now complete.

\subsection{Proof of lemma \ref{lbunderlineOmegalossy}}
\label{prooflbunderlineOmegalossy}

\subsubsection{Removing the Dependence on the Identification Index}

Recall from Section \ref{sec:premnonasymp} that for each $w$, the joint distribution of $(X^n(w),Y^n(w),S(w),Z^n,\hatX^n)$ is $P_{\bX\bY S\bZ\mathbf{\hatX}}$ and $P_{X_iY_iZ_iU_iV_i}$ is induced by $P_{S\bX\bY\bZ\mathbf{\hatX}}$. Furthermore, recall that $Q_i$ denotes $Q_{X_iY_iZ_iU_i\hatX_i}$ and $P_i$ denotes $P_{X_iY_iZ_iU_iV_i}$. Define
\begin{align}
\nn&g_{Q_i,P_i}^{(\alpha,\mu,\beta,\lambda)}(x_i,y_i,z_i,\hatx_i|u_i,v_i)\\*
&:=\Bigg(\frac{1}{f^{(\alpha,\mu,\beta)}_{Q_i,P_i}(x_i,y_i,z_i,\hatx_i|u_i,v_i)}\Bigg)^{\lambda}\label{flambda}.
\end{align}

Invoking \eqref{def:omegaamblambda}, we obtain \eqref{lossyOmeganow} on the top of  the next page,
\begin{figure*}
\begin{align}
\exp\Big(-\Omega^{(\alpha,\mu,\lambda)}(\{Q_i\}_{i=1}^n)\Big)
\nn&=\sum_{w=1}^M\frac{1}{M}\Bigg(\sum_{x^n(w),y^n(w),s(w),z^n,\hatx^n}P_{\bX\bY S\bZ\mathbf{\hatX}}(x^n(w),y^n(w),s(w),z^n,\hatx^n)\\*
&\qquad\qquad\qquad\qquad \prod_{i=1}^n  g_{Q_i,P_i}^{(\alpha,\mu,\beta,\lambda)}(x_i(w),y_i(w),z_i,\hatx_i|u_i(w),v_i(w))\Bigg)\\
&=\sum_{x^n,y^n,s,z^n,\hatx^n}P_{\bX\bY S\bZ\mathbf{\hatX}}
(x^n,y^n,s,z^n,\hatx^n)\prod_{i=1}^n g_{Q_i,P_i}^{(\alpha,\mu,\beta,\lambda)}(x_i,y_i,z_i,\hatx_i|u_i,v_i)
\label{lossyOmeganow},
\end{align}
\hrulefill
\end{figure*}
where $u_i=(s,y^{i-1},z_{i+1}^n)$, $v_i=(s,z_{i+1}^n)$ and the joint distribution of $X^n,Y^n,S,Z^n,\hatX^n$ is $P_{\bX\bY S\bZ\mathbf{\hatX}}$. Let $U_i=(S,Y^{i-1},Z_{i+1}^n)$. Then we have the Markov chains $Z_i-X_i-Y_i-U_i$ and $(X_i,Y_i)-(U_i,Z_i)-\hatX_i$. In the following, all the distributions are induced by $P_{\bX\bY S\bZ\mathbf{\hatX}}$ and we will omit subscripts of distributions for convenience.

\subsubsection{Preliminaries} 

Invoking \eqref{lossyOmeganow}, we obtain that
\begin{align}
\nn&\exp\Big(-\Omega^{(\alpha,\mu,\beta,\lambda)}(\{Q_i\}_{i=1}^n)\Big)\\*
\nn&=\sum_{s,z^n}P(s,z^n)\sum_{x^n,y^n,\hatx^n}P(x^n,y^n,\hatx^n|s,z^n)\\*
&\qquad\times\prod_{i=1}^ng_{Q_i,P_i}^{(\alpha,\mu,\beta,\lambda)}(x_i,y_i,z_i,\hatx_i|u_i,v_i)\label{lossyOmegaow2}.
\end{align}
For $i=1,\ldots,n$, define
\begin{align}
\nn&C_i(s,z^n)
:=\sum_{x^i,y^i,\hatx^n}P(x^i,y^i,\hatx^i|s,z^n)\\*
&\qquad\qquad\qquad\times\prod_{j=1}^i g_{Q_j,P_j}^{(\alpha,\mu,\beta,\lambda)}(x_j,y_j,z_j,\hatx_j|u_j,v_j)\label{def:Cisz}\\
\nn&P^{(\alpha,\mu,\beta,\lambda)}(x^i,y^i,\hatx^i|s,z^n):=\frac{P(x^i,y^i,\hatx^i|s,y^i)}{C_i(s,z^n)}\\*
&\qquad\qquad\qquad\times\prod_{j=1}^i g_{Q_j,P_j}^{(\alpha,\mu,\beta,\lambda)}(x_j,y_j,z_j,\hatx_j|u_j,v_j),\label{def:Pamblxyhatxgszi}\\
&\Psi_i^{(\alpha,\mu,\beta,\lambda)}(s,z^n|\{Q_j\}_{j=1}^i):=C_i(s,z^n)/C_{i-1}(s,z^n)\label{def:Phisz}.
\end{align}

Similarly as \cite[Lemma 6]{oohama2016wynerziv}, we obtain the following lemma.
\begin{lemma}
\label{propPhisz}
For $i=1,\ldots,n$ and any $(s,z^n,x^t,y^t,\hatx^t)$, we have
\begin{align}
\nn&\Psi_i^{(\alpha,\mu,\beta,\lambda)}(s,z^n|\{Q_j\}_{j=1}^i)\\*
\nn&=\sum_{x^i,y^i,\hatx^i}P^{(\alpha,\mu,\beta,\lambda)}(x^{i-1},y^{i-1},\hatx^{i-1}|s,z^n)\\*
\nn&\qquad\times P(x_i,y_i,\hatx_i|s,x^{i-1},y^{i-1},\hatx^{i-1},z^n)\\*
&\qquad\times g_{Q_i,P_i}^{(\alpha,\mu,\beta,\lambda)}(x_i,y_i,z_i,\hatx_i|u_i,v_i).
\end{align}
\end{lemma}
Hence, combining \eqref{lossyOmegaow2} and Lemma \ref{propPhisz}, we obtain that
\begin{align}
\nn&\exp\Big(-\Omega^{(\alpha,\mu,\beta,\lambda)}(\{Q_i\}_{i=1}^n)\Big)\\*
&=\sum_{s,z^n}P(s,z^n)\prod_{i=1}^n \Psi_i^{(\alpha,\mu,\beta,\lambda)}(s,z^n|\{Q_j\}_{j=1}^i)\label{lossyOmegaow3}.
\end{align}
Then, for $i=1,\ldots,n$, define
\begin{align}
&\tilC_i:=\sum_{s,z^n} P(s,z^n)\prod_{j=1}^i\Psi_j^{(\alpha,\mu,\beta,\lambda)}(s,z^n|\{Q_l\}_{l=1}^j),\label{def:tilCi}\\
&P_{SZ^n}^{(\alpha,\mu,\beta,\lambda)|i}(s,z^n):=\frac{P(s,z^n)}{\tilC_i}\prod_{j=1}^i \Psi_j^{(\alpha,\mu,\beta,\lambda)}(s,z^n|\{Q_l\}_{l=1}^j),\label{def:PSZiambl}\\
&\Lambda_i^{(\alpha,\mu,\beta,\lambda)}(\{Q_j\}_{j=1}^i):=\tilC_i/\tilC_{i-1}.\label{def:Lambdai}
\end{align}
Similarly as \cite[Lemma 7]{oohama2016wynerziv}, we obtain the following lemma.
\begin{lemma}
\label{propLambdai}
For $i=1,\ldots,n$, we have
\begin{align}
\nn\Lambda_i^{(\alpha,\mu,\beta,\lambda)}(\{Q_j\}_{j=1}^i)
&=\sum_{s,z^n}P_{SZ^n}^{(\alpha,\mu,\beta,\lambda)|i-1}(s,z^n)\\*
&\quad\times\Psi_i^{(\alpha,\mu,\beta,\lambda)}(s,z^n|\{Q_j\}_{j=1}^i).
\end{align}
\end{lemma}

Using \eqref{lossyOmegaow3} and Lemma \ref{propLambdai}, we obtain
\begin{align}
\exp\Big(-\Omega^{(\alpha,\mu,\beta,\lambda)}(\{Q_i\}_{i=1}^n)\Big)=\prod_{i=1}^n\Lambda_i^{(\alpha,\mu,\beta,\lambda)}(\{Q_j\}_{j=1}^i).
\label{lossyOmegaow4}
\end{align}

\subsubsection{Final proof of Lemma \ref{lbunderlineOmegalossy}}

Recall \eqref{def:omegaamblambdaQP}. Define
\begin{align}
\hat{\calQ}_{n}
\nn&:=\Big\{Q_{XYZU\hatX}:\\*
&\quad|\calU|\leq |\calL| |\calX|^{n-1}|\calY|^{n-1}|\hat{\calX}|^{n-1}|\calZ|^{n-1}\Big\}\label{def:lossyhatcalQn},\\
\hat{\Omega}_{n}^{(\alpha,\mu,\beta,\lambda)}
&:=\min_{Q_{XYZU\hatX}\in\hat{\calQ}_{n}}\Omega^{(\alpha,\mu,\beta,\lambda)}(Q_{XYZU\hatX})\label{def:hatOmeganambl}.
\end{align}

Recall that $u_i=(s,y^{i-1},z_{i+1}^n)$. For each $i=1,\ldots,n$, define
\begin{align}
\nn&P^{(\alpha,\mu,\beta,\lambda)}(s,x_i,y^i,z_i^n,\hatx_i)=P^{(\alpha,\mu,\beta,\lambda)}(x_i,y_i,z_i,u_i,\hatx_i)\\
\nn&:=\sum_{x^{i-1},z^{i-1},\hatx^{i-1}} \Big(P^{(\alpha,\mu,\beta,\lambda)|i-1}(s,z^n)\\*
\nn&\qquad\qquad\times P^{(\alpha,\mu,\beta,\lambda)}(x^{i-1},y^{i-1},\hatx^{i-1}|s,z^n)\\
&\qquad\qquad\times P(x_i,y_i,\hatx_i|s,x^{i-1},y^{i-1},\hatx^{i-1},z^n)\Big)\label{def:PSXYZhXambl},
\end{align}
where $P^{(\alpha,\mu,\beta,\lambda)}(x^{i-1},y^{i-1},\hatx^{i-1}|s,z^n)$ was defined in \eqref{def:Pamblxyhatxgszi} and $P^{(\alpha,\mu,\beta,\lambda)|i-1}(s,z^n)$ was defined in \eqref{def:PSZiambl}.

Invoking Lemmas \ref{propPhisz}, \ref{propLambdai} and \eqref{def:PSXYZhXambl}, we obtain that for $i=1,\ldots,n$,
\begin{align}
\nn&\Lambda_i^{(\alpha,\mu,\beta,\lambda)}(\{Q_j\}_{j=1}^i)\\*
\nn&=\sum_{x_i,y_i,z_i,u_i,\hatx_i}P^{(\alpha,\mu,\beta,\lambda)}(x_i,y_i,z_i,u_i,\hatx_i)\\*
&\qquad\times g_{Q_i,P_i}^{(\alpha,\mu,\beta,\lambda)}(x_i,y_i,z_i,\hatx_i|u_i,v_i)\label{Lambdaieq}.
\end{align}

Note that $Q_i=Q_{X_iY_iZ_iU_i\hatX_i}$ can be chosen arbitrarily for all $i=1,\ldots,n$. Here we apply the recursive method. For each $i=1,\ldots,n$, we choose $Q_{X_iY_iZ_iU_i\hatX_i}$ such that
\begin{align}
\nn&Q_{X_iY_iZ_iU_i\hatX_i}(x_i,y_i,z_i,u_i,\hatx_i)\\*&=P^{(\alpha,\mu,\beta,\lambda)}(x_i,y_i,z_i,u_i,\hatx_i)\label{tochooseqandp}.
\end{align}
Then, let $Q_{Y_i}$, $Q_{Z_i}$, $Q_{X_i|Y_iZ_iU_i}$, $Q_{X_iY_i|Z_iU_i\hatX_i}$, $Q_{Y_iZ_i|U_i}$, $Q_{Y_iZ_i}$, $Q_{X_iY_i|Z_iU_i}$, $Q_{Z_i|U_i}$ be induced by $Q_{X_iY_iZ_iU_i\hatX_i}$. Thus, we have $Q_{X_iY_iZ_iU_i\hatX_i}\in\hat{\calQ}_{n}$. 

Using the definition of $g^{\cdot}_{\cdot}(\cdot)$ in \eqref{flambda}, define
\begin{align}
\nn&h^{(\alpha,\mu,\beta,\lambda)}_{Q_i,P_{X_iY_iZ_i}}(x_i,y_i,z_i,\hatx_i|u_i)\\*
\nn&:=g^{(\alpha,\mu,\beta,\lambda)}_{Q_i,P_i}(x_i,y_i,z_i,\hatx_i|u_i,v_i)\\
&\quad\times\Bigg(\frac{P^{\lambda}_{X_iY_i|Z_iU_i}(x_i,y_i|z_i,u_i)}{Q^{\lambda}_{X_iY_i|Z_iU_i}(x_i,y_i|z_i,u_i)}\frac{P_{Z_i|V_i}^{\lambda\alpha\bar{\mu}\beta}(z_i|v_i)}{Q_{Z_i|U_i}^{\lambda\alpha\bar{\mu}\beta}(z_i|u_i)}\Bigg)^{-1}\label{def:fpxyzi}.
\end{align}
Recall that $\bar{a}=1-a$ for $a\in[0,1]$. In the following, for simplicity, we will drop the subscripts of the distributions. Furthermore, let $\psi:=1-\lambda-\lambda\alpha\bar{\mu}\beta$. From \eqref{Lambdaieq} and \eqref{tochooseqandp}, we obtain
\begin{align}
\nn&\Lambda_i^{(\alpha,\mu,\beta,\lambda)}(\{Q_j\}_{j=1}^i)\\*
\nn&=\mathbb{E}_{\{Q_j\}_{j=1}^i}\big[
g^{(\alpha,\mu,\beta,\lambda)}_{Q_i,P_i}(X_i,Y_i,Z_i,\hatX_i|U_i,V_i)
\big]\\
\nn&=\mathbb{E}_{\{Q_j\}_{j=1}^i}\Bigg[
h^{(\alpha,\mu,\beta,\lambda)}_{Q_i,P_{X_iY_iZ_i}}(X_i,Y_i,Z_i,\hatX_i|U_i)\\*
&\qquad\qquad\qquad\times \frac{P^{\lambda}(X_i,Y_i|Z_i,U_i)}{Q^{\lambda}(X_i,Y_i|Z_i,U_i)}\frac{P^{\lambda\alpha\bar{\mu}\beta}(Z_i|V_i)}{Q^{\lambda\alpha\bar{\mu}\beta}(Z_i|U_i)}
\Bigg]\label{usefcompare}\\
\nn&\leq \Bigg(\mathbb{E}_{\{Q_j\}_{j=1}^i}\Bigg[
\Bigg\{h^{(\alpha,\mu,\beta,\lambda)}_{Q_i,P_{X_iY_iZ_i}}(X_i,Y_i,Z_i,\hatX_i|U_i)\Bigg\}^{\frac{1}{\psi}}\Bigg]\Bigg)^{\psi}\\
\nn&\qquad\times \Bigg(\mathbb{E}_{\{Q_j\}_{j=1}^i}\Bigg[\frac{P(X_i,Y_i|Z_i,U_i)}{Q(X_i,Y_i|Z_i,U_i)}\Bigg]\Bigg)^{\lambda}\\*
&\qquad\times\Bigg(\mathbb{E}_{\{Q_j\}_{j=1}^i}\Bigg[\frac{P(Z_i|V_i)}{Q(Z_i|U_i)}\Bigg]\Bigg)^{\lambda\alpha\bar{\mu}\beta}\label{useholderineq}\\
&=\exp\bigg(\!-\Big(1\!-\lambda\!-\lambda\alpha\bar{\mu}\beta\Big)\Omega^{(\alpha,\mu,\beta,\frac{\lambda}{1\!-\lambda-\lambda\alpha\bar{\mu}\beta})}(Q_i)\bigg)\label{usedeflossyOmegaambl}\\
&=\exp\Bigg(-\frac{\Omega^{(\alpha,\mu,\beta,\theta)}(Q_i)}{1+\theta+\theta\alpha\bar{\mu}\beta}\Bigg)\label{usetheta}\\
&\leq \exp\Bigg(-\frac{\hat{\Omega}_{n}^{(\alpha,\mu,\beta,\theta)}}{1+\theta+\theta\alpha\bar{\mu}\beta}\Bigg)\label{usehOmeganambl}\\*
&=\exp\Bigg(-\frac{\Omega^{(\alpha,\mu,\beta,\theta)})}{1+\theta+\theta\alpha\bar{\mu}\beta}\Bigg)\label{uselemmacard},
\end{align}
where \eqref{usefcompare} follows from \eqref{def:fpxyzi}; \eqref{useholderineq} follows from H\"{o}lder's inequality; \eqref{usedeflossyOmegaambl} follows from the definitions in \eqref{def:omegaamblambdaQP} and \eqref{def:fpxyzi}; \eqref{usetheta} follows from \eqref{def:theta} and \eqref{lambdausetheta}; \eqref{usehOmeganambl} follows since $Q_{XYZU\hatX}^*\in\hat{\calQ}_n$ (recall \eqref{def:lossyhatcalQn}); and \eqref{uselemmacard} follows from since the cardinality bound $|\calU|\leq |\calX||\calY||\calZ||\hat{\calX}|$ is sufficient to describe $\hat{\Omega}_n^{(\alpha,\beta,\mu,\lambda)}$ (the proof of this   is similar to~\cite[Property~4(a)]{oohama2016wynerziv} and thus omitted).

Hence, combining \eqref{lossyOmegaow4} and \eqref{uselemmacard}, we obtain that
\begin{align}
\frac{1}{n}\Omega^{(\alpha,\mu,\beta,\lambda)}(\{Q_i\}_{i=1}^n)
&=-\frac{1}{n}\sum_{i=1}^n\log \Lambda_i^{(\alpha,\mu,\beta,\lambda)}(\{Q_j\}_{j=1}^i)\\*
&\geq \frac{\Omega^{(\alpha,\mu,\beta,\theta)})}{1+\theta+\theta\alpha\bar{\mu}\beta}\label{finallyready}.
\end{align}
Finally, combining~\eqref{def:lossyunderlineOmega} and \eqref{finallyready}, we conclude that
\begin{align}
\underline{\Omega}^{(\alpha,\mu,\beta,\lambda)}&\geq \frac{\Omega^{(\alpha,\mu,\beta,\theta)})}{1+\theta+\theta\alpha\bar{\mu}\beta}.
\end{align}
The proof of Lemma \ref{lbunderlineOmegalossy} is now complete.

\vspace{-0.05in}
\subsection{Proof of Lemma \ref{propFOmega}}
\label{proofpropFOmega}

Before proceeding the proof of Lemma \ref{propFOmega}, we need the following definitions.
Let 
\begin{align}
\calP
\nn&:=\Big\{Q_{XYZU\hatX}:~|\calU|\leq |\calY|+2,~Z-X-Y-U\\*
\nn&\qquad\qquad Q_X=P_X,~Q_{Y|X}=P_{Y|X},~Q_{Z|X}=P_{Z|X},\\*
&\qquad\qquad (X,Y)-(U,Z)-\hatX\Big\}\label{def:calPlossy},\\
\calR_{\mathrm{ran}}&:=\bigcup_{Q_{XYZU\hatX}\in\calP}\calR(Q_{XYZU\hatX})\label{def:calRran}.
\end{align}
\blue{
Furthermore, let
\begin{align}
\calP_{\mathrm{sh}}
\nn&:=\Big\{Q_{XYZU\hatX}\in\calP(\calX\times\calY\times\calZ\times\calU\times \hat{\calX}):|\calU|\leq |\calY|\\*
\nn&\quad\qquad Q_X=P_X,~Q_{Y|X}=P_{Y|X},~Q_{Z|X}=P_{Z|X},\\*
&\qquad\quad Z-X-Y-U,~(X,Y)-(U,Z)-\hatX\Big\}.\label{def:calPshlossy}
\end{align}
}
\vspace{-0.05in}

Recall that given a number $a\in[0,1]$, we define $\bar{a}=1-a$. Then for any $(\mu,\beta)\in\times[0,1]^2$, define
\begin{align}
R^{(\mu,\beta)}
\nn&:=\min_{P_{XYZU\hatX}\in\calP_{\mathrm{sh}}}
\Big\{ \bar{\mu}\bar{\beta}I(P_{YZ},P_{U|YZ})\\*
&\qquad\quad-\bar{\mu}I(P_Z,P_{Z|U})+\mu \mathbb{E}_{P_{X\hatX}}[d(X,\hatX)]\Big\},\label{def:lossyRmubeta}\\
\calR_{\mathrm{sh}}
\nn&:=\bigcap_{(\mu,\beta)\in[0,1]^2}\Big\{(R^{\rmi},R^{\rmc},D):\bar{\mu}\bar{\beta}R^{\rmc}-\bar{\mu}R^{\rmi}+\mu D\\*
&\qquad\qquad\qquad\qquad\qquad\qquad\geq R^{(\mu,\beta)}\Big\}.\label{def:calRshlossy}
\end{align}

Similarly as \cite[Property 3]{oohama2016wynerziv}, we can prove the following lemma, which plays an important role in the proof of Lemma~\ref{propFOmega}.
\begin{lemma}
Recalling the definition of $\calR^*$ in \eqref{def:calRstar}, we have
\begin{align}
\calR_{\mathrm{sh}}=\calR^*=\calR=\calR_{\rm{ran}}\label{rsheqrran}.
\end{align}
\end{lemma}

\subsubsection{Proof of Conclusion i)}
\blue{
Similar as \eqref{def:omegaamb} and \eqref{def:omegaamblambdaQP}, for each $P_{XYZU\hatX}\in\calP_{\rm{sh}}$ and any $\lambda\in\bbR_+$, let
\begin{align}
\nn&\tilde{\omega}^{(\mu,\beta)}_{P_{XYZU\hatX}}(x,y,z,u,\hatx):=\bar{\mu}\bar{\beta}\log\frac{P_{YZ|U}(y,z|u)}{P_{YZ}(y,z)}\\*
&\qquad\qquad\qquad\qquad+\bar{\mu}\beta\log\frac{P_Z(z)}{P_{Z|U}(z|u)}+\mu d(x,\hatx),\label{def:tildeomegamb}\\
\nn&\tilde{\Omega}^{(\lambda,\mu,\beta)}(P_{XYZU\hatX})\\*
&\!\!\!:=\!-\log \mathbb{E}_{P_{XYZU\hatX}}[\exp(-\lambda \tilde{\omega}^{(\mu,\beta)}_{P_{XYZU\hatX}}\!\!(X,Y,Z,U,\hatX))]\label{def:tildeOmegambtP}.
\end{align}
Similarly as \eqref{def:omegaamblambdaP}, \eqref{def:lossyFalphamubetalambda} and \eqref{def:lossyF}, define
\begin{align}
&\tilde{\Omega}^{(\lambda,\mu,\beta)}
:=\min_{P_{XYZU\hatX}\in\calP_{\rm{sh}}}\tilde{\Omega}^{(\lambda,\mu,\beta)}(P_{XYZU\hatX}),\label{def:tildeOmegambt}\\
\nn&\tilF^{(\lambda,\mu,\beta)}(R^\rmi,R^\rmc,D)\\*&:=\frac{\tilde{\Omega}^{(\lambda,\beta,\mu)}-\lambda(\bar{\mu}(\bar{\beta}R^\rmc-R^\rmi)+\mu D)}{6+\lambda\bar{\mu}(4+6\beta)},\label{def:tilFmbt}\\
\nn&\tilF(R^\rmi,R^\rmc,D)\\*&:=\sup_{(\lambda,\beta,\mu)\in\bbR_+\times [0,1]^2}\tilF^{(\lambda,\mu,\beta)}(R^\rmi,R^\rmc,D).\label{def:tilF}
\end{align}
For any $P_{XYZU\hatX}\in\calP_{\rm{sh}}$, define the   tilted distribution for any $(\lambda,\beta,\mu)\in\bbR_+\times[0,1]^2$ as in \eqref{tilpxzuhx} and define the parameter $\rho$ as in \eqref{def:rho} both on the top of  the next page.
\begin{figure*}
\begin{align}
P_{XYZU\hatX}^{(\lambda,\mu,\beta)}(x,y,z,u,\hatx)
&:=\frac{P_{XYZU\hatX}(x,y,z,u,\hatx)\exp(-\lambda\tilde{\omega}^{(\mu,\beta)}_{P_{XYZU\hatX}}(x,y,z,u,\hatx))}{\mathbb{E}_{P_{XYZU\hatX}}[\exp(-\lambda\tilde{\omega}^{(\mu,\beta)}_{P_{XYZU\hatX}}(X,Y,Z,U,\hatX))]}\label{tilpxzuhx},\\
\rho&:=\sup_{P_{XYZU\hatX}\in\calP_{\rm{sh}}}\sup_{\substack{(\lambda,\mu,\beta)\in\bbR_+\times[0,1]^2:\\\lambda\bar{\mu}\leq 1}}\mathrm{Var}_{P_{XYZU\hatX}^{(\lambda,\mu,\beta)}}\Big[\tilde{\omega}^{(\mu,\beta)}_{P_{XYZU\hatX}}(X,Y,Z,U,\hatX)\Big]\label{def:rho}.
\end{align}
\hrulefill
\end{figure*}
Note that $\rho$ is positive and finite.
}

\blue{
We then have the following lemma.
\begin{lemma}
\label{propcalRshlossy}
The following hold:
\begin{itemize}
\item[i)] For any $(R^\rmi,R^\rmc,D)$, we have
\begin{align}
F(R^\rmi,R^\rmc,D)\geq \tilF(R^\rmi,R^\rmc,D).
\end{align}
\item[ii)] If $(R^\rmi,R^\rmc,D)\notin\calR$, we have that for some $\delta\in(0,\rho)$,
\begin{align}
\tilF(R^\rmi,R^\rmc,D)>\frac{\delta^2}{8\rho}>0.
\end{align}
\end{itemize}
\end{lemma}
We remark that   ii) in Lemma \ref{propFOmega} follows directly from Lemma~\ref{propcalRshlossy}. We now present the proof of Lemma \ref{propcalRshlossy}, which follows along the lines of~\cite{oohama2016new,oohama2016wynerziv}.
}
\begin{proof}[Proof of Lemma \ref{propcalRshlossy}]
\blue{
For any $Q_{XYZU\hatX}\in\calQ$ (cf. \eqref{def:calQlossy}),  let $P_{XYZU\hatX}\in\calP_{\rm{sh}}$ (cf. \eqref{def:calPshlossy}) be chosen such that $P_{U|Y}=Q_{U|Y}$ and $P_{\hatX|ZU}=Q_{\hatX|ZU}$. Now, using the definition of $\Omega^{(\alpha,\mu,\beta,\theta)}(Q_{XYZU\hatX})$ in \eqref{def:omegaamblambdaQP}, for any $(\alpha,\theta,\mu,\beta)\in\bbR_+^2\times[0,1]^2$ such that $\theta(1+\alpha\bar{\mu})\leq 1$ and $\alpha\bar{\mu}\beta\leq 1$, we obtain \eqref{reason5} on the top of next page,
\begin{figure*}
\begin{align}
\nn&\exp\{-\Omega^{(\alpha,\mu,\beta,\theta)}(Q_{XYZU\hatX})\}\\
\nn&=\mathbb{E}_{Q_{XYZU\hatX}}\Bigg[\bigg(\frac{P_{XYZ}(X,Y,Z)Q_{XY|ZU}(X,Y|Z,U)}{Q_Y(Y)Q_{Z|YU}(Z|Y,U)Q_{X|YZU}(X|Y,Z,U)Q_{XY|ZU\hatX}(X,Y|Z,U,\hatX)}\bigg)^{\theta}\\
&\qquad\times \bigg(\frac{P_{YZ}(Y,Z)}{Q_{YZ|U}(Y,Z|U)}\bigg)^{\theta\alpha\bar{\mu}\bar{\beta}}\bigg(\frac{Q_{Z|U}(Z|U)}{Q_Z(Z)}\bigg)^{\theta\alpha\bar{\mu}\beta}\exp\big(-\theta\alpha\mu d(X,\hatX)\big)\Bigg]\\
&=\mathbb{E}_{Q_{XYZU\hatX}}\Bigg[\bigg(\frac{P_{XYZU\hatX}(X,Y,Z,U,\hatX)}{Q_{XYZU\hatX}(X,Y,Z,U,\hatX)} \bigg(\frac{P_{YZ}(Y,Z)}{Q_{YZ|U}(Y,Z |U)}\bigg)^{\alpha\bar{\mu}\bar{\beta}}\bigg(\frac{Q_{Z|U}(Z|U)}{Q_Z(Z)}\bigg)^{\alpha\bar{\mu}\beta}\exp\big(-\alpha\mu d(X,\hatX)\big)\bigg)^{\theta}\Bigg]\label{reason1}\\
\nn&=\mathbb{E}_{Q_{XYZU\hatX}}\Bigg[\bigg(\frac{P_{XYZU\hatX}(X,Y,Z,U,\hatX)}{Q_{XYZU\hatX}(X,Y,Z,U,\hatX)} \bigg(\frac{P_{YZ}(Y,Z)}{P_{YZ|U}(Y,Z|U)}\bigg)^{\alpha\bar{\mu}\bar{\beta}}\bigg(\frac{Q_{Z|U}(Z|U)}{P_Z(Z)}\bigg)^{\alpha\bar{\mu}\beta}\exp\big(-\alpha\mu d(X,\hatX)\big)\bigg)^{\theta}\\*
&\qquad\qquad\qquad\times \bigg(\frac{P_{YZ|U}(Y,Z|U)}{Q_{YZ|U}(Y,Z|U)}\bigg)^{\theta\alpha\bar{\mu}\bar{\beta}}\bigg(\frac{P_Z(Z)}{Q_Z(Z)}\bigg)^{\theta\alpha\bar{\mu}\beta}\Bigg]\\
\nn&\leq \Bigg(\mathbb{E}_{Q_{XYZU\hatX}}\Bigg[\frac{P_{XYZU\hatX}(X,Y,Z,U,\hatX)}{Q_{XYZU\hatX}(X,Y,Z,U,\hatX)} \frac{P_{YZ}^{\alpha\bar{\mu}\bar{\beta}}(Y,Z)}{P_{YZ|U}^{\alpha\bar{\mu}\bar{\beta}}(Y,Z|U)}\frac{Q_{Z|U}^{\alpha\bar{\mu}\beta}(Z|U)}{P_Z^{\alpha\bar{\mu}\beta}(Z)}\exp\big(-\alpha\mu d(X,\hatX)\big)\Bigg]\Bigg)^{\theta}\\
&\qquad\times \Bigg(\mathbb{E}_{Q_{XYZU\hatX}}\Bigg[\bigg(\frac{P_{YZ|U}(Y,Z|U)}{Q_{YZ|U}(Y,Z|U)}\bigg)^{\theta\alpha\bar{\mu}/\bar{\theta}}\Bigg]\Bigg)^{\bar{\beta}\bar{\theta}}\Bigg(\mathbb{E}_{Q_{XYZU\hatX}}\Bigg[\bigg(\frac{P_Z(Z)}{Q_Z(Z)}\bigg)^{\theta\alpha\bar{\mu}/\bar{\theta}}\Bigg]\Bigg)^{\beta\bar{\theta}}\label{reason2}\\
&\leq \Bigg(\mathbb{E}_{P_{XYZU\hatX}}\Bigg[\frac{P_{YZ}^{\alpha\bar{\mu}\bar{\beta}}(Y,Z)}{P_{YZ|U}^{\alpha\bar{\mu}\bar{\beta}}(Y,Z|U)}\frac{Q_{Z|U}^{\alpha\bar{\mu}\beta}(Z|U)}{P_Z^{\alpha\bar{\mu}\beta}(Z)}\exp\big(-\alpha\mu d(X,\hatX)\big)\Bigg]\Bigg)^{\theta}\label{reason3}\\
&=\Bigg(\mathbb{E}_{P_{XYZU\hatX}}\Bigg[\frac{P_{YZ}^{\alpha\bar{\mu}\bar{\beta}}(Y,Z)}{P_{YZ|U}^{\alpha\bar{\mu}\bar{\beta}}(Y,Z|U)}\frac{P_{Z|U}^{\alpha\bar{\mu}\beta}(Z|U)}{P_Z^{\alpha\bar{\mu}\beta}(Z)}\exp\big(-\alpha\mu d(X,\hatX)\big)  \bigg(\frac{Q_{Z|U}(Z|U)}{P_{Z|U}(Z|U)}\bigg)^{\alpha\bar{\mu}\beta}\Bigg]\Bigg)^{\theta}\\
&\leq \Bigg(\!\mathbb{E}_{P_{XYZU\hatX}}\Bigg[\bigg(\frac{P_{YZ}^{\alpha\bar{\mu}\bar{\beta}}(Y,Z)}{P_{YZ|U}^{\alpha\bar{\mu}\bar{\beta}}(Y,Z|U)}\frac{P_{Z|U}^{\alpha\bar{\mu}\beta}(Z|U)}{P_Z^{\alpha\bar{\mu}\beta}(Z)}\exp\big(\!-\!\alpha\mu d(X,\hatX)\big)\bigg)^{\frac{1}{1-\alpha\bar{\mu}\beta}}\Bigg]\Bigg)^{\theta(1-\alpha\bar{\mu}\beta)} \Bigg(\!\mathbb{E}_{P_{XYZU\hatX}}\bigg[\frac{Q_{Z|U}(Z|U)}{P_{Z|U}(Z|U)}\bigg]\!\Bigg)^{\alpha\bar{\mu}\beta}\label{reason4}\\
&=\exp\Big\{-\theta(1-\alpha\bar{\mu}\beta)\tilde{\Omega}^{(\frac{\alpha}{1-\alpha\bar{\mu}\beta},\beta,\mu)}(P_{XYZU\hatX})\Big\}\label{reason5}.
\end{align}
\hrulefill
\end{figure*}
where \eqref{reason1} follows since i) the choice of $P_{XYZU\hatX}\in\calP_{\rm{sh}}$ satisfies that $P_{XYZ} P_{U|Y} P_{\hatX|UZ} =P_{XYZU\hatX}$ and ii) the equality 
\begin{equation}
 \frac{Q_{XY|ZU} }{Q_{XY|ZU\hatX} }= \frac{Q_{\hatX|ZU} }{Q_{\hatX|XYZU} };
\end{equation}
\eqref{reason2} follows from H\"older's inequality; \eqref{reason3} follows since $\mathbb{E}[X^a]$ is concave in $X$ for $a\in[0,1]$ and the choice of $\theta$ (recall that $\theta\in[0,\frac{1}{1+\alpha}]$); \eqref{reason4} follows from  H\"older's inequality; and \eqref{reason5} follows from the definition in \eqref{def:tildeOmegambtP}.
}

\blue{
Using \eqref{reason5} and the definitions in \eqref{def:omegaamblambdaP} and \eqref{def:tildeOmegambt}, we conclude that for any $(\alpha,\theta,\mu,\beta)\in\bbR_+^2\times[0,1]^2$ such that $\theta(1+\alpha\bar{\mu})\leq 1$ and $\alpha\bar{\mu}\beta\leq 1$,
\begin{align}
\Omega^{(\alpha,\mu,\beta,\theta)}
&\geq \theta(1-\alpha\bar{\mu}\beta)\tilde{\Omega}^{(\frac{\alpha}{1-\alpha\bar{\mu}\beta},\beta,\mu)}.
\end{align}
Using the definitions in \eqref{def:lossyF} and \eqref{def:tilF}, we have \eqref{reason9} on the top of the page after next,
\begin{figure*}
\begin{align}
F(R^\rmi,R^\rmc,D)
&=\sup_{(\alpha,\theta,\mu,\beta)\in\bbR_+^2\times[0,1]^2}\frac{\Omega^{(\alpha,\mu,\beta,\theta)}-\theta\alpha(\bar{\mu}(\bar{\beta}R^\rmc-R^\rmi)+\mu D)}{1+5\theta+\theta\alpha\bar{\mu}(3-\beta)}\\
&\geq \sup_{\substack{(\alpha,\theta,\mu,\beta)\in\bbR_+^2\times[0,1]^2:\\\theta(1+\alpha\bar{\mu})\leq 1,~\alpha\bar{\mu}\beta\leq 1}}\frac{\theta(1-\alpha\bar{\mu}\beta)\tilde{\Omega}^{(\frac{\alpha}{1-\alpha\bar{\mu}\beta},\beta,\mu)}-\theta\alpha(\bar{\mu}(\bar{\beta}R^\rmc-R^\rmi)+\mu D)}{1+5\theta+\theta\alpha\bar{\mu}(3-\beta)}\\
&=\sup_{\substack{(\alpha,\mu,\beta)\in\bbR_+\times[0,1]^2:\\\alpha\bar{\mu}\beta\leq 1}}\sup_{\substack{\theta\in\bbR_+: \\ \theta(1+\alpha\bar{\mu})\leq 1}}\frac{\theta \big((1-\alpha\bar{\mu}\beta)\tilde{\Omega}^{(\frac{\alpha}{1-\alpha\bar{\mu}\beta},\beta,\mu)}-\alpha(\bar{\mu}(\bar{\beta}R^\rmc-R^\rmi)+\mu D)\big)}{1+5\theta+\theta\alpha\bar{\mu}(3-\beta)}\\
&\geq \sup_{\substack{(\alpha,\mu,\beta)\in\bbR_+\times[0,1]^2:\\\alpha\bar{\mu}\beta\leq 1}}\frac{(1-\alpha\bar{\mu}\beta)\tilde{\Omega}^{(\frac{\alpha}{1-\alpha\bar{\mu}\beta},\beta,\mu)}-\alpha(\bar{\mu}(\bar{\beta}R^\rmc-R^\rmi)+\mu D)}{6+4\alpha\bar{\mu}}\label{reason7}\\
&=\sup_{(\lambda,\mu,\beta)\in\bbR_+\times[0,1]^2}\frac{\tilde{\Omega}^{(\lambda,\beta,\mu)}-\lambda (\bar{\mu}(\bar{\beta}R^\rmc-R^\rmi)+\mu D)}{6+\lambda\bar{\mu}(4+6\beta)}\label{reason8}\\
&=\tilF(R^\rmi,R^\rmc,D)\label{reason9}.
\end{align}
\hrulefill
\end{figure*}
where \eqref{reason7} follows since 
\begin{align}
\frac{\theta}{1+5\theta+\theta\alpha\bar{\mu}(3-\beta)}
&\geq \frac{\theta}{1+5\theta+3\theta\alpha\bar{\mu}},\\
\sup_{\theta\in\bbR_+:\theta(1+\alpha\bar{\mu})\leq 1}\frac{\theta}{1+5\theta+3\theta\alpha\bar{\mu}}&=\frac{1}{6+4\alpha\bar{\mu}};
\end{align}
\eqref{reason8} follows by choosing $\lambda=\frac{\alpha}{1-\alpha\bar{\mu}\beta}$ and noting that $\alpha\bar{\mu}\beta\leq 1$ implies that $\lambda\in\bbR_+$; and \eqref{reason9} follows from the definition in \eqref{def:tilF}.
}

\blue{
In the following, we will show that $\tilF(R^\rmi,R^\rmc,D)>0$ for any triple $(R^\rmi,R^\rmc,D)$ such that $(R^\rmi,R^\rmc,D)\notin\calR$. Using the definition in \eqref{def:tildeOmegambtP}, we conclude that for any $P_{XYZU\hatX}\in\calP_{\rm{sh}}$,
\begin{align}
\nn&\frac{\partial \tilde{\Omega}^{(\lambda,\mu,\beta)}(P_{XYZU\hatX})}{\partial \lambda}\\*
&=\mathbb{E}_{P_{XYZU\hatX}^{(\lambda,\mu,\beta)}}\Big[\tilde{\omega}^{(\mu,\beta)}_{P_{XYZU\hatX}}(X,Y,Z,U,\hatX)\Big]\label{td1},\\
\nn&\frac{\partial^2 \tilde{\Omega}^{(\lambda,\mu,\beta)}(P_{XYZU\hatX})}{\partial \lambda^2}\\*
&=-\mathrm{Var}_{P_{XYZU\hatX}^{(\lambda,\mu,\beta)}}\Big[\tilde{\omega}^{(\mu,\beta)}_{P_{XYZU\hatX}}(X,Y,Z,U,\hatX)\Big]\label{td2}.
\end{align}
Applying a Taylor expansion to $\tilde{\Omega}^{(\lambda,\mu,\beta)}(P_{XYZU\hatX})$ around $\lambda=0$, we obtain that for any $P_{XYZU\hatX}\in\calP_{\rm{sh}}$ and any $\lambda\in[0,\frac{1}{\bar{\mu}}]$, there exists some $\tau\in[0,\lambda]$ such that 
\begin{align}
\nn&\tilde{\Omega}^{(\lambda,\mu,\beta)}(P_{XYZU\hatX})\\*
\nn&\!=\!\tilde{\Omega}^{(0,\mu,\beta)}(P_{XYZU\hatX})\!+\!\lambda \mathbb{E}_{P_{XYZU\hatX}^{(0,\mu,\beta)}}\Big[\tilde{\omega}^{(\mu,\beta)}_{P_{XYZU\hatX}}(X,Y,Z,U,\hatX)\Big]\\
&\qquad-\frac{\lambda^2}{2}\mathrm{Var}_{P_{XYZU\hatX}^{(\tau,\mu,\beta)}}\Big[\tilde{\omega}^{(\mu,\beta)}_{P_{XYZU\hatX}}(X,Y,Z,U,\hatX)\Big]\label{treason1}\\
&\geq \lambda \mathbb{E}_{P_{XYZU\hatX}}\Big[\tilde{\omega}^{(\mu,\beta)}_{P_{XYZU\hatX}}(X,Y,Z,U,\hatX)\Big]-\frac{\lambda^2\rho}{2}\label{treason2}.
\end{align}
where \eqref{treason1} follows from \eqref{td1} and \eqref{td2}; and \eqref{treason2} follows from \eqref{def:rho}.}

\blue{
Using the definitions in \eqref{def:lossyRmubeta}, \eqref{def:tildeOmegambt}
and the result in \eqref{treason2}, we obtain
\begin{align}
\tilde{\Omega}^{(\lambda,\mu,\beta)}
&=\min_{P_{XYZU\hatX}\in\calP_{\rm{sh}}}\tilde{\Omega}^{(\lambda,\mu,\beta)}(P_{XYZU\hatX})\\
&\geq \lambda R^{(\mu,\beta)}-\frac{\lambda^2\rho}{2}.\label{result2}
\end{align}
Now consider any triple $(R^\rmi,R^\rmc,D)$ outside the first-order coding region, i.e., $(R^\rmi,R^\rmc,D)\notin\calR$ (cf. \eqref{def:calRstar}). Invoking conclusion i) in Lemma \ref{propcalRshlossy}, we conclude that there exists $\mu^*\in[0,1]$ and $\beta^*\in[0,1]$ such that for some positive $\delta\in(0,\rho]$ (cf. \eqref{def:rho})
\begin{align}
\bar{\mu}^*\bar{\beta}^*R^\rmc-\bar{\mu}^*R^\rmi+\mu^* D\leq R^{(\mu^*,\beta^*)}-\delta\label{result3}.
\end{align}
Using the definition in \eqref{def:tilF}, we obtain that for all  $\lambda \in [0, \frac{1}{\bar{\mu}}]$.
\begin{align}
\nn&\tilF(R^\rmi,R^\rmc,D)\\*
&=\sup_{\substack{(\lambda,\beta,\mu)\\\in\bbR_+\times[0,1]^2}}\!\!\!\! \frac{\tilde{\Omega}^{(\lambda,\beta,\mu)}-\lambda(\bar{\mu}(\bar{\beta}R^\rmc-R^\rmi)+\mu D)}{6+\lambda\bar{\mu}(4+6\beta)}\\
&\geq \sup_{\lambda\in\bbR_+} \frac{\tilde{\Omega}^{(\lambda,\beta^*,\mu^*)}-\lambda(\bar{\mu}^*(\bar{\beta}^*R^\rmc-R^\rmi)+\mu^* D)}{6+\lambda\bar{\mu}^*(4+6\beta^*)}\\
&\geq \sup_{\lambda\in[0,1]} \frac{\lambda\delta-\frac{\lambda^2\rho}{2}}{6+10\lambda} \label{freason1}\\
&\geq \sup_{\lambda\in[0,1]}\frac{1}{16}\bigg(-\frac{\rho}{2}\Big(\lambda-\frac{\delta}{\rho}\Big)^2+\frac{2\delta^2}{\rho}\bigg)\label{freason2}\\
&=\frac{\delta^2}{8\rho},\label{freason3}
\end{align}
where \eqref{freason1} follows from \eqref{result2}, \eqref{result3} and noting that $\lambda\leq 1$ implies that $\lambda\leq \frac{1}{\bar{\mu}^*}$ since $\mu^*\in[0,1]$; \eqref{freason2} follows since $6+10\lambda\leq 16$ for $\lambda\in[0,1]$; and \eqref{freason3} follows since $\delta\leq \rho$.}
\end{proof}

\vspace{-0.08in}
\subsubsection{Proof of Conclusion ii)}
In the following, we will present a proof for conclusion ii) in Lemma \ref{propFOmega}. If $(R^{\rmi},R^{\rmc},D)\in\calR$, then there exists a joint distribution $Q_{XYZU\hatX}^*\in\calP$ such that
\begin{align}
R^{\rmc}-R^{\rmi}&\geq I(Q_{U|Z}^*,Q^*_{Y|UZ}|Q^*_Z),\\
R^{\rmi}&\leq I(Q^*_Z,Q^*_{U|Z}), \\
D &\geq \mathbb{E}_{Q^*_{X\hatX}}[d(X,\hatX)].
\end{align}
Hence, for any $(\alpha,\mu,\beta)$, we have
\vspace{-0.02in}
\begin{align}
\nn&\bar{\mu}\bar{\beta}R^{\rmc}-\bar{\mu}R^{\rmi}+\mu D\\*
&=\bar{\mu}\bar{\beta}(R^{\rmc}-R^{\rmi})-\bar{\mu}\beta R^{\rmi}+\mu D\\
\nn&\geq \bar{\mu}\bar{\beta}\Big(I(Q_{YZ}^*,Q_{U|YZ}^*)-I(Q_Z^*,Q_{Z|U}^*)\Big)\\*
&\qquad-\bar{\mu}\beta I(Q_Z^*,Q^*_{U|Z})+\mu\mathbb{E}_{Q^*_{X\hatX}}[d(X,\hatX)]\\
\nn&=\bar{\mu}\bar{\beta}I(Q_{YZ}^*,Q_{U|YZ}^*)-\bar{\mu}I(Q_Z^*,Q_{U|Z}^*)\\*
&\qquad+\mu\mathbb{E}_{Q^*_{X\hatX}}[d(X,\hatX)].
\end{align}
Applying Taylor expansions to $\Omega^{(\alpha,\mu,\beta,\theta)}(Q_{XYZU\hatX})$, we conclude that
\begin{align}
\nn&\Omega^{(\alpha,\mu,\beta,\theta)}(Q_{XYZU\hatX})\\*
&\leq \theta\mathbb{E}_{Q_{XYZU\hatX}}\Big[\omega^{(\alpha,\mu,\beta)}_{Q_{XYZU\hatX}}(x,y,z,\hatx|u)\Big]\label{readyuse2}.
\end{align}
Combining \eqref{def:calPlossy} and \eqref{def:calQlossy}, we conclude that $\calQ\supseteq\calP$. Thus, invoking \eqref{def:omegaamblambdaP} and \eqref{readyuse2}, we have
\begin{align}
\nn&\Omega^{(\alpha,\mu,\beta,\theta)}(Q_{XYZU\hatX})\\*
&=\min_{Q_{XYZU\hatX}\in\calQ}
\Omega^{(\alpha,\mu,\beta,\theta)}(Q_{XYZU\hatX})\\
&\leq \min_{Q_{XYZU\hatX}\in\calP}\theta\mathbb{E}_{Q_{XYZU\hatX}}\Big[\omega^{(\alpha,\mu,\beta)}_{Q_{XYZU\hatX}}(x,y,z,\hatx|u)\Big]\\
\nn&=\min_{Q_{XYZU\hatX}\in\calP}
\theta\alpha \Big(\bar{\mu}\bar{\beta}I(Q_{YZ},Q_{U|YZ})-\bar{\mu}I(Q_Z,Q_{U|Z})\\*
&\qquad\qquad\qquad\qquad+\mu\mathbb{E}_{Q_{X\hatX}}[d(X,\hatX)]\Big)\\
\nn&\leq \theta\alpha\Big(\bar{\mu}\bar{\beta}I(Q_{YZ}^*,Q_{U|YZ}^*)-\bar{\mu}I(Q_Z^*,Q_{U|Z}^*)\\*
&\qquad\qquad+\mu\mathbb{E}_{Q^*_{X\hatX}}[d(X,\hatX)]\Big)\\
&\leq \theta\alpha\Big(\bar{\mu}\bar{\beta}R^{\rmc}-\bar{\mu}R^{\rmi}+\mu D\Big)\label{readyuse3}.
\end{align}
Thus, combining \eqref{def:lossyFalphamubetalambda} and \eqref{readyuse3}, we obtain that
\begin{align}
F^{(\alpha,\mu,\beta,\theta)}
&=\frac{\Omega^{(\alpha,\mu,\beta,\theta)}-\theta\alpha\Big(\bar{\mu}(\bar{\beta} R^{\rmc}-R^{\rmi})+\mu D\Big)}{1+5\theta+\theta\alpha\bar{\mu}(3-\beta)}\\*
&\leq 0\label{fnonpositive}.
\end{align}
On the other hand, note that
\begin{align}
\lim_{\theta\to 0}F^{(\alpha,\mu,\beta,\theta)}=0\label{fnonnegetive}.
\end{align}
Hence, combining \eqref{fnonpositive} and \eqref{fnonnegetive}, we conclude that
\begin{align}
F
&=\sup_{(\alpha,\theta,\mu,\beta)\in\bbR_+^2\times [0,1]^2}F^{(\alpha,\mu,\beta,\theta)}=0.
\end{align}

\subsection{Proof of  the Extensions for the Biometrical Identification Problem}
\label{proofbioext}

\subsubsection{Exponent of the Probability of Correct Decoding}

Specializing Lemma \ref{fblidlossy} to the biometrical problem (using $\calA_5(w)$ only), we obtain that for any decoding function $g^{(n)}$ and any $\eta\geq 0$,
\begin{align}
\rmP_\rmc^{(n)}(g^{n})
\nn&\leq \Pr\Bigg\{\frac{1}{n}\sum_{i=1}^n \log\frac{P_{Z|Y}(Z_i|Y_i)}{P_Z(Z_i)}\geq R^\rmi-\eta\Bigg\}\\*
&\qquad+\exp(-n\eta)\label{bio2}.
\end{align}
Furthermore, adopting the one-shot technique in \cite{yassaee2013technique}, we conclude that there exists a decoding function $g^{(n)}$ and $\gamma\geq 0$ such that 
\begin{align}
\rmP_\rmc^{(n)}(g^{n})
\nn&\geq \frac{1}{1+\exp(-n\gamma)}\\*
&\times\Pr\Bigg\{\frac{1}{n}\sum_{i=1}^n \log\frac{P_{Z|Y}(Z_i|Y_i)}{P_Z(Z_i)}\geq R^\rmi+\gamma\Bigg\}.\label{bio1}
\end{align}

Due to the memoryless of the source and channel, we have that $(X_i,Y_i,Z_i)$ is an i.i.d. sequence, distributed according to $P_X\times P_{Y|X}\times P_{Z|X}$. Specializing \eqref{bio1} with $\gamma=0$ and using Cram\'er's theorem~\cite[Theorem 2.2.3]{dembo2009large}, we obtain   
\begin{align}
\nn&\rmP_\rmc^{(n)}(g^{n})\\*
\nn&\geq \frac{1}{2}\exp\Bigg(-n \sup_{\lambda>0}\Bigg\{\lambda R^\rmi\\*
&\qquad\qquad\quad-\log \mathbb{E}\left[\exp\left(\lambda\log\frac{P_{Y|Z}(Y|Z)}{P_Z(Z)}\right)\right]\Bigg\}\Bigg)\\
&=\frac{1}{2}\exp\Bigg(-n\sup_{\lambda>0}\Bigg\{\lambda R^\rmi-\log\mathbb{E}\left[\log\frac{P_{Z|Y}^\lambda(Z|Y)}{P_Z^\lambda(Z)}\right]\Bigg\}\Bigg).
\end{align}
Combining~\eqref{bio2} and Lemma \ref{cramer}, 
\begin{align}
\nn&\rmP_\rmc^{(n)}(g^{n})\\*
\nn&\leq\exp(-n\eta)+\exp\Bigg(-n\lambda(R^\rmi-\eta)\\*
&\qquad\quad+\log\mathbb{E}\left[\exp\left(\lambda\sum_{i=1}^n \log\frac{P_{Z|Y}(Z_i|Y_i)}{P_Z(Z_i)}\right)\right]\Bigg)\\
\nn&=\exp(-n\eta)+\exp\Bigg(-n\lambda(R^\rmi-\eta)\\*
&\qquad\qquad+n\log\mathbb{E}\left[\exp\left(\lambda \log\frac{P_{Z|Y}(Z|Y)}{P_Z(Z)}\right)\right]\Bigg)\label{iidxyz},
\end{align}
where \eqref{iidxyz} follows since $(X_i,Y_i,Z_i)$ is an i.i.d.\ sequence.

Choose $\eta$ such that
\begin{align}
\eta=\lambda (R^\rmi-\eta)-\log\mathbb{E}\left[\exp\left(\lambda \log\frac{P_{Z|Y}(Z|Y)}{P_Z(Z)}\right)\right].
\end{align}
In other words,
\begin{align}
\eta=\frac{\lambda R^\rmi-\log\mathbb{E}\big[\exp \big(\lambda \log\frac{P_{Z|Y}(Z|Y)}{P_Z(Z)}\big)\big]}{1+\lambda}.
\end{align}
With this choice of $\eta$, we obtain that 
\begin{align}
&\rmP_\rmc^{(n)}(g^{n})\leq 2\exp(-n\eta)\\*
&\! =\!  2\exp\Bigg(  \! -\!  n\sup_{\lambda>0}\! \frac{\lambda R^\rmi\! -\! \log\mathbb{E}\big[\exp\big(\lambda \log\frac{P_{Z|Y}(Z|Y)}{P_Z(Z)}\big)\big]}{1+\lambda}\Bigg).
\end{align}

\subsubsection{Moderate Deviations Constant in the Strong Converse Regime \eqref{theorem:mdcsc}}
\label{mdcsc}

Let
\begin{align}
nR^\rmi=\log M:=nI(P_Y,P_{Z|Y})+n\xi_n.
\end{align}

Invoking \eqref{bio1} and choosing $\gamma=\zeta\xi_n$ for some $\zeta>0$, we conclude that there exists a sequence of decoding function $g^{(n)}$ such that 
\begin{align}
\rmP_\rmc^{(n)}(g^{n})\nn
\nn&\geq \frac{1}{1+\exp(-n\zeta\xi_n)}\Pr\Bigg\{\frac{1}{n}\sum_{i=1}^n \log\frac{P_{Z_i|Y_i}(Z_i|Y_i)}{P_Z(Z_i)}\\*
&\qquad\qquad\quad\geq I(P_Y,P_{Z|Y})+(1+\zeta)\xi_n\Bigg\}.
\end{align}
Using the moderate deviations theorem~\cite[Theorem~3.7.1]{dembo2009large}, we obtain  
\begin{align}
&\lim_{n\to\infty}- \frac{1}{n\xi_n^2} \log \Pr\bigg\{\frac{1}{n}\sum_{i=1}^n \log\frac{P_{Z|Y}(Z_i|Y_i)}{P_Z(Z_i)} \nn\\*
&\quad\geq I(P_Y,P_{Z|Y})+(1+\zeta)\xi_n\bigg\}  =\frac{(1+\zeta)^2}{2\rmV}\label{usemdc}.
\end{align}
Therefore, we have 
\begin{align}
\liminf_{n\to\infty}\frac{-\log \rmP_\rmc^{(n)}(g^{(n)})}{n\xi_n^2}
&\le\frac{(1+\zeta)^2}{2\rmV}\label{useusemdc}.
\end{align}

On the other hand, for any decoding function and any $M$ such that
\begin{align}
nR^\rmi=\log M=I(P_Y,P_{Z|Y})+n\xi_n,
\end{align}
invoking \eqref{bio2} and choosing $\eta=(1+\zeta)\xi_n$, we obtain that
\begin{align}
\rmP_\rmc^{(n)}(g^{n})
\nn&\leq \exp(-n\zeta\xi_n)+\Pr\Bigg\{\frac{1}{n}\sum_{i=1}^n \log\frac{P_{Z|Y}(Z_i|Y_i)}{P_Z(Z_i)}\\*
&\qquad\qquad\quad\geq I(P_Y,P_{Z|Y})+(1-\zeta)\xi_n\Bigg\}\label{mdcwhodomi}.
\end{align}
Similar as \eqref{usemdc}, we conclude that the first term in \eqref{mdcwhodomi} is of the order $\exp(-n\xi_n^2\frac{(1-\zeta)^2}{2\rmV})$ and thus dominates the right hand side of \eqref{mdcwhodomi} for sufficiently large $n$. Hence,
\begin{align}
\limsup_{n\to\infty}\frac{-\log \rmP_\rmc^{(n)}(g^{(n)})}{n\xi_n^2}
&\geq\frac{(1-\zeta)^2}{2\rmV}.
\end{align}
The proof is complete by letting $\zeta\to 0$.

\subsubsection{Moderate Deviations Constant \eqref{theorem:mdc}}
Invoking \eqref{bio2}, we obtain that for any decoding function $g^{(n)}$, we have
\begin{align}
\rmP_\rme^{(n)}(g^{(n)})
&=1-\rmP_\rmc^{(n)}(g^{(n)})\\
\nn&\geq \Pr\Bigg\{\frac{1}{n}\sum_{i=1}^n \log\frac{P_{Z|Y}(Z_i|Y_i)}{P_Z(Z_i)}< R^\rmi-\eta\Bigg\}\\*
&\qquad-\exp(-n\eta)\label{bio22}.
\end{align}

Invoking \eqref{bio1}, we obtain that there exists a decoding function $g^{(n)}$ such that 
\begin{align}
\nn&\rmP_\rme^{(n)}(g^{(n)})\\
&\leq 1-\frac{\Pr\left\{\frac{1}{n}\sum_{i=1}^n \log\frac{P_{Z|Y}(Z_i|Y_i)}{P_Z(Z_i)}\geq R^\rmi+\gamma\right\}}{1+\exp(-n\gamma)}\\
&\nn\leq \Pr\Bigg\{\frac{1}{n}\sum_{i=1}^n \log\frac{P_{Z|Y}(Z_i|Y_i)}{P_Z(Z_i)}< R^\rmi+\gamma\Bigg\}\\*
&\qquad+\left(1-\frac{1}{1+\exp(-n\gamma)}\right)\\
\nn&\leq \Pr\Bigg\{\frac{1}{n}\sum_{i=1}^n \log\frac{P_{Z|Y}(Z_i|Y_i)}{P_Z(Z_i)}< R^\rmi+\gamma\Bigg\}\\*
&\qquad+\exp(-n\gamma)\label{bio11}.
\end{align}

The rest of the proof is similar to that in Appendix \ref{mdcsc} by invoking \eqref{bio22}, \eqref{bio11} with properly chosen  $\gamma$ and $\eta$ as well as applying the moderate deviations theorem.

\subsubsection{Second-order Asymptotics}: The result in Theorem \ref{bioext3} follows by i) letting $\gamma=\eta=\frac{\log n}{n}$ and ii) applying  the Berry-Esseen theorem to \eqref{bio2} and \eqref{bio1} or to \eqref{bio22} and \eqref{bio11}.

\subsection*{Acknowledgments} 
The authors would like to thank Prof.\ Yasutada Oohama for providing updated versions of his manuscripts~\cite{oohama2016new,oohama2016wynerziv,oohama2015wak}.

\bibliographystyle{IEEEtran}
\bibliography{IEEEfull_lin}

\end{document}